\documentclass[a4paper,10pt,DIV=10]{scrartcl}

\usepackage[style=english]{csquotes}
\usepackage{booktabs}
\usepackage{float}
\usepackage{mathtools}
\usepackage{wasysym}
\usepackage{relsize}
\usepackage{bm}
\usepackage{xcolor}
\usepackage{graphicx}
\usepackage{stackengine}
\usepackage{tikz}
\usetikzlibrary{patterns}

\usepackage{microtype} 
\usepackage{amsmath,amssymb,amsthm}
\definecolor{myblue}{HTML}{3D6288}
\definecolor{myred}{HTML}{9F1D2B}
\definecolor{mygreen}{HTML}{5B892D}
\usepackage{hyperref}
\hypersetup{%
	breaklinks=true,%
	pdfencoding=auto,%
	colorlinks,%
	linkcolor=myred,%
	citecolor=mygreen,%
	urlcolor=myblue,%
}
\usepackage[capitalise]{cleveref}

\makeatletter
\renewcommand{\subparagraph}{%
	\@startsection {subparagraph}{5}{\z@ }{3.25ex \@plus 1ex
	\@minus .2ex}{-1em}{\normalfont \normalsize \bfseries\sffamily }}%
\makeatother

\usepackage{authblk}
\title{Infinite Probabilistic Databases}

\author[1]{Martin Grohe}
\affil[1]{\bgroup\large\url{grohe@informatik.rwth-aachen.de}\egroup}

\author[2]{Peter Lindner}
\affil[2]{\bgroup\large\url{lindner@informatik.rwth-aachen.de}\egroup}

\affil[\bgroup\empty\egroup]{\bgroup\large RWTH Aachen University\egroup}

\DeclareSymbolFont{sfoperators}{OT1}{cmss}{m}{n}
\SetSymbolFont{sfoperators}{bold}{OT1}{cmss}{b}{n}
\makeatletter
\renewcommand{\operator@font}{\mathgroup\symsfoperators}
\makeatother

\let\cal\mathcal
\let\bb\mathbb

\let\fr\mathfrak

\let\sf\mathsf 
\let\epsilon\varepsilon 
\let\phi\varphi 
\let\rho\varrho

\newcommand*{\N}{\bb N}
\newcommand*{\R}{\bb R}
\newcommand*{\Q}{\bb Q}

\DeclarePairedDelimiter{\length}{\lvert}{\rvert}
\DeclarePairedDelimiter{\card}{\lvert}{\rvert}
\DeclarePairedDelimiter{\size}{\lvert}{\rvert}
\renewcommand*{\complement}[1]{{#1}^{\mathsf{c}}}
\newcommand*{\hits}{\mathord{\#}}

\newcommand*{\Iff}{\quad\Leftrightarrow\quad}

\renewcommand*{\quad}[1][1]{\kern#1em}

\DeclarePairedDelimiter{\set}{\{}{\}}
\def\Sets#1#2{%
	\left\lparen%
	\mathchoice{\mkern-2mu}{\mkern-2mu}{\mkern-3mu}{\mkern-2.5mu}%
	\begin{smallmatrix}#1\\#2\end{smallmatrix}%
	\mathchoice{\mkern-2mu}{\mkern-2mu}{\mkern-3mu}{\mkern-2.5mu}%
	\right\rparen%
}
\DeclarePairedDelimiterX{\dparens}[1]{\lparen}{\rparen}{%
	\nparen{\lparen}{#1}\delimsize\lparen\mathopen{}%
	#1%
	\mathclose{}\delimsize\rparen\nparen{\rparen}{#1}%
}
\DeclarePairedDelimiterX{\dbraces}[1]{\lbrace}{\rbrace}{%
	\nbrace{\lbrace}{#1}\delimsize\lbrace\mathopen{}%
	#1%
	\mathclose{}\delimsize\rbrace\nbrace{\rbrace}{#1}%
}
\newcommand{\dummydelim}[2]{$\left#1\vphantom{#2}\right.$}
\newcommand{\nparen}[2]{\sbox0{\dummydelim{#1}{#2}}\hspace{-0.5\wd0}}
\newcommand{\nbrace}[2]{\sbox0{\dummydelim{#1}{#2}}\hspace{\the\dimexpr -0.85\wd0 + 2pt\relax}}

\makeatletter
\def\bag{\@ifstar\@bag\@@bag}
\def\@bag#1{\dbraces{\smash{#1}}}
\def\@@bag#1{\dbraces*{#1}}

\newcommand{\Bags}[3][\@nil]{%
	\def\tmp{#1}%
	\ifx\tmp\@@nil%
		\dparens*{\smash{\begin{smallmatrix}#2\\#3\end{smallmatrix}}}%
	\else%
		\dparens*{\begin{smallmatrix}#2\\#3\end{smallmatrix}}%
	\fi%
}
\makeatother

\DeclareMathOperator{\proj}{proj}
\DeclareMathOperator{\sym}{sym}

\newcommand*{\Attributes}{\bm{\sf{Attributes}}}
\newcommand*{\Relations}{\bm{\sf{Relations}}}

\newcommand*{\Schema}[1][S]{\cal{#1}}
\newcommand*{\Att}[1][A]{\cal{#1}}
\newcommand*{\Rel}[1][R]{\cal{#1}}

\DeclareMathOperator{\ar}{ar}
\DeclareMathOperator{\dom}{dom}
\DeclareMathOperator{\facts}{facts}
\DeclareMathOperator{\type}{type}

\newcommand*{\Topology}[1][T]{\fr{#1}}
\newcommand*{\Measurable}[1][D]{\fr{#1}}
\newcommand*{\Borel}{{\fr{B\mkern-1muo\mkern-1mur}}}

\newcommand*{\Space}[1][D]{\bb{#1}}
\newcommand*{\Instances}{\Space[D]}

\newcommand*{\Event}[1][D]{\cal{#1}}
\newcommand*{\CEvent}{\Event[C]}

\newcommand*{\pdb}{\Delta}
\newcommand*{\bbDelta}{\ensurestackMath{%
	\stackinset{r}{0.025em}{b}{0em}{\scalebox{.75}{$\Delta$}}{\Delta}}%
}
\newcommand*{\PDBs}{\bbDelta}

\newcommand*{\select}[1][]{\sigma_{#1}}
\newcommand*{\project}[1][]{\pi_{#1}}
\newcommand*{\product}{\times}
\newcommand*{\addunion}{\uplus}
\newcommand*{\maxunion}{\cup}
\newcommand*{\intersection}{\cap}
\newcommand*{\difference}{-}
\newcommand*{\dedupe}{\delta}
\newcommand*{\rename}[1][]{\rho_{#1}}
\newcommand*{\natjoin}{\mathbin{\raisebox{-1pt}{\normalfont$\Bowtie$}}}

\newcommand*{\aggregatorop}[1]{%
	\mathord{\text{\normalfont\sffamily\bfseries\smaller #1}}%
}
\newcommand*{\CNT}{\aggregatorop{CNT}}
\newcommand*{\CNTd}{\aggregatorop{CNTd}}
\newcommand*{\SUM}{\aggregatorop{SUM}}
\newcommand*{\AVG}{\aggregatorop{AVG}}
\newcommand*{\MAX}{\aggregatorop{MAX}}
\newcommand*{\MIN}{\aggregatorop{MIN}}

\theoremstyle{plain}
\newtheorem{theorem}{Theorem} %
\newtheorem{lemma}[theorem]{Lemma} %
\newtheorem{corollary}[theorem]{Corollary} %
\newtheorem{claim}[theorem]{Claim} %
\crefname{claim}{Claim}{Claims} %
\newtheorem{fact}[theorem]{Fact}
\crefname{fact}{Fact}{Facts}

\theoremstyle{remark}
\newtheorem{example}[theorem]{Example} %
\newtheorem{remark}[theorem]{Remark} %
\newtheorem{definition}[theorem]{Definition} %
\newtheorem{observation}[theorem]{Observation}
\crefname{observation}{Observation}{Observations}

\newcommand*{\sketchname}{Proof Sketch}
\newenvironment{sketch}[1][\sketchname]%
	{\begin{proof}[#1]}{\end{proof}}
\newenvironment{subproof}[1][\proofname]{%
		\begin{proof}[#1]%
	}{\end{proof}}

\usepackage[textwidth=2.7cm,textsize=scriptsize,shadow]{todonotes}
\colorlet{mgcolor}{red!40}
\colorlet{plcolor}{blue!40}

\newcommand{\atmg}{{\setlength{\fboxsep}{2pt}\colorbox{black!65!white}{\color{mgcolor}\bfseries\sffamily@MG}}}
\newcommand{\atpl}{{\setlength{\fboxsep}{2pt}\colorbox{black!65!white}{\color{plcolor}\bfseries\sffamily@PL}}}

\begin{document}

\maketitle

\begin{abstract}
Probabilistic databases (PDBs) are used to model uncertainty in data in a
quantitative way. In the standard formal framework, PDBs are finite
probability spaces over relational database instances. It has been argued
convincingly that this is not compatible with an open world semantics (Ceylan
et al., KR~2016) and with application scenarios that are modeled by continuous
probability distributions (Dalvi et al., CACM~2009).\par

We recently introduced a model of PDBs as infinite probability spaces that
addresses these issues (Grohe and Lindner, PODS~2019). While that work was
mainly concerned with countably infinite probability spaces, our focus here is
on uncountable spaces. Such an extension is necessary to model typical
continuous probability distributions that appear in many applications. However,
an extension beyond countable probability spaces raises nontrivial foundational
issues concerned with the measurability of events and queries and
ultimately with the question whether queries have a well-defined
semantics.

It turns out that so-called finite point processes are the appropriate model
from probability theory for dealing with probabilistic databases. This model
allows us to construct suitable (uncountable) probability spaces of database
instances in a systematic way. Our main technical results are measurability
statements for relational algebra queries as well as aggregate queries and 
datalog queries.
\end{abstract}

\section{Introduction}
Probabilistic databases (PDBs) are used to model uncertainty in data.
Such uncertainty could be introduced by a variety of reasons
like, for example, noisy sensor data, the presence of incomplete or
inconsistent information, or because the information is gathered from
unreliable sources \cite{Aggarwal+2009,Suciu+2011}. In the standard
formal framework, probabilistic databases are finite
probability spaces whose sample spaces consist of database instances
in the usual sense, referred to as \enquote{possible worlds}. However,
this framework has various shortcomings due to its inherent
\emph{closed world assumption} \cite{Ceylan+2016}---in particular, any
event outside of the finite scope of such probabilistic databases is
treated as an impossible event. There is also work on PDBs that
includes continuous probability distributions and hence goes beyond the
formal framework of finite probability space. Yet, these continuous PDBs lack 
a general formal basis in terms of a possible-worlds semantics 
\cite{Dalvi+2009}. While both open world PDBs and continuous
probability distributions in PDBs have received some attention in the
literature, there is no systematic joint treatment of these issues with a sound
theoretical foundation. In \cite{Grohe+2019}, we introduced an extended model 
of PDBs as arbitrary (possibly infinite) probability spaces over finite
database instances. However, the focus there was on countably
infinite PDBs. An extension to continuous PDBs, which is necessary to
model probability distributions appearing in many applications that
involve real-valued measurement data, raises new fundamental questions
concerning the measurability of events and queries.\par\medskip

In this paper, we lay the foundations of a systematic and sound
treatment of infinite, even uncountable, probabilistic databases, and we
prove that queries expressed in standard query languages have a
well-defined semantics.\par
Our treatment is based on the mathematical
theory of finite point processes \cite{Moyal1962,Macchi1975,Daley+2003}. 
Adopting this theory to the context of relational databases, we give
a suitable construction of measurable spaces over which our
probabilistic databases can then be defined. The only assumption that we
need to make is that the domains of all attributes satisfy certain
topological assumptions (they need to be Polish spaces; all standard
domains such as integers, strings, reals, satisfy this assumption). For queries and views to
have a well-defined open world semantics, we need them to be measurable mappings
between probabilistic databases. Our main technical result states that
indeed all queries and views that can be expressed in the relational algebra,
even equipped with arbitrary aggregate operators (satisfying some
mild measurability conditions) are measurable mappings. The result
holds for both a bag-based and set-based relational algebra.  We also prove
the measurability of datalog queries.

Measurability of queries may seem like an obvious minimum requirement,
but one needs to be very careful. We give an example of a simple, innocent
looking ``query'' that is not measurable (see Example~\ref{ex:nonborel}). The
proofs of the measurability results are not trivial, which may already
be seen from the fact that they depend on the topological assumption
that the attribute domains are Polish spaces (most importantly, they
are complete topological spaces and have a countable dense
subset). At their core, the proofs are based on finding suitable
``countable approximations'' of the queries.

In the last section of this paper, we briefly discuss queries for
probabilistic databases that go beyond ``standard''
database queries lifted to probabilistic databases via an open
world-semantics. Examples of such a queries are probabilistic threshold
queries and rank queries. Such queries refer not only to the facts in a
database, but also to their probabilities, and hence are inherently probabilistic.

\subparagraph*{Related Work}
Early work on models for probabilistic databases dates back to the 1980s
\cite{Wong1982,Gelenbe+1986,Cavallo+1987} and 1990s \cite{Barbara+1992,
Dey+1996,Fuhr+1997,Zimanyi1997}. These models may be seen as special cases or
variations of the now-acclaimed formal model of probabilistic databases that
features a usually finite set of database instances (the \enquote{possible
worlds}) together with a probability distribution among them
\cite{Aggarwal+2009,Suciu+2011}.\par

The work \cite{Koch2008} presents a formal definition of the probabilistic
semantics of relational algebra queries as it is used in the MayBMS system
\cite{Koch2009}.
A probabilistic semantics for datalog has already been proposed in the mid-90s
\cite{Fuhr1995}. More recently, a version of datalog was considered in which
rules may fire probabilistically \cite{Deutch+2010}. Aggregate queries in
probabilistic databases were first treated systematically in \cite{Ross+2005}
and reappear in various works concerning particular PDB systems \cite{Murthy+2011,Fink+2012}.
\par

The models of possible worlds semantics mentioned above are the mathematical
backbone of existing probabilistic database prototype systems such as MayBMS
\cite{Koch2009}, Trio \cite{Widom2009} and MystiQ \cite{Boulos+2005}.
Various subsequent prototypes feature uncountable domains as well, such as
Orion \cite{Singh+2008}, MCDB \cite{Jampani+2008,Jampani+2011}, new versions of 
Trio \cite{Agrawal+2009} and PIP \cite{Kennedy+2010}. The MCDB system in 
particular allows programmers to specify probabilistic databases with 
infinitely many possible worlds with database instances that can grow 
arbitrarily large \cite{Jampani+2011} and is therefore probably the most
general existing system. Its system-driven description does not feature a 
general formal, measure theoretic account of its semantics though. In a spirit 
that is similar to our presentation here, the work \cite{Tran+2012}
introduced a measure theoretic semantics for probabilistic data stream systems 
with probability measures composed from Gaussian mixture models but (to our
knowledge) on a per tuple basis and without the possibility of inter-tuple 
correlations. Continuous probabilistic databases have already been considered 
earlier in the context of sensor networks \cite{Faradjian+2002,Cheng+2003,Deshpande+2004}. 
The first work to formally introduce continuous possible worlds semantics
(including aggregation) is \cite{Abiteboul+2011} for probabilistic XML.
However, the framework has an implicit restriction bounding the number of
tuples in a PDB.\par

Models similar in expressivity to the one we present have also been suggested 
in the context of probabilistic modeling languages and probabilistic
programming \cite{Milch+2005,Milch2006,Richardson+2006,DeRaedt+2016,Barany+2017}.
In particular notable are the measure theoretic treatments of Bayesian Logic
(BLOG) \cite{Milch+2005} in \cite{Wu+2018} and Markov Logic Networks (MLNs)
\cite{Richardson+2006} in \cite{Singla+2007}. While these data models are
relational, it is unclear, how suitable they are for general database 
applications and in particular, the investigation of typical database queries
is beyond the scope of these works.\par

Problems raised by the closed world assumption \cite{Reiter1978} in
probabilistic databases was discussed initially by Ceylan et al. in
\cite{Ceylan+2016} where they suggest the model of OpenPDBs. In
\cite{Borgwardt+2018}, the authors make a more fine-grained distinction between
an \emph{open world} and \emph{open domain assumption}, the latter of which
does not assume the attribute values of the database schema to come from a
known finite domain. The work \cite{Friedman+2019} considers semantic
constraints on open worlds in the OpenPDB framework. The semantics of OpenPDBs 
can be strengthened towards an open domain assumption by the means of 
ontologies \cite{Borgwardt+2017,Borgwardt+2018,Borgwardt+2019}.\par

The classification of views we discuss towards the end of this paper shares
similarities with previous classifications of queries such as \cite{Cheng+2003}
in the sense that it distinguishes \emph{how} aggregation is involved. The work
\cite{Wanders+2015} suggests a distinction between \enquote{traditional} and
\enquote{out-of-world aggregation} quite similar to the one we present.\par\medskip

\section{Preliminaries}\label{sec:preliminaries}
Throughout the paper, we denote the set of nonnegative integers by $\N$, the
set of rational numbers by $\Q$ and the set of real numbers by $\R$. We write
$\N_+$, $\Q_+$ and $\R_+$ for the restrictions of these sets to strictly 
positive numbers.\par

If $M$ is a set and $k\in\N$, then $\Sets{M}{k}$ denotes the set of subsets of
$M$ of cardinality $k$. The set of all finite subsets of $M$ is then given by
$\bigcup_{k\geq 0}\Sets{M}{k}\eqqcolon \Sets{M}{<\omega}$.\par
A \emph{bag} (also called \emph{multiset}) over a set $U$ is an unordered
collection of elements of $U$, possibly with repetitions. In order to
distinguish sets and bags, we use double curly braces $\bag{\cdots}$ when
explicitly denoting bags. Similarly to the notation for sets, we let
$\Bags{M}{k}$ denote the set of bags over the set $M$ of cardinality $k\in\N$
(that is, containing $k$ elements, counting copies). The set of all finite bags 
over $M$ is given by $\bigcup_{k\geq 0}\Bags{M}{k}\eqqcolon \Bags{M}{<\omega}$.
\par

There are multiple equivalent ways to formalize the notion of bags. We
introduce two such definitions that we use interchangeably later:

\begin{description}
	\item[Multiplicity perspective] A \emph{bag} $B$ over some 
		set $U$ is a function $\hits_B\colon U\to\N$ assigning a 
		\emph{multiplicity} to every element of $U$. The cardinality of 
		$B$ is $\card{B}\coloneqq\sum_{u\in U}\hits_B(u)$.
	\item[Quotient perspective] For all $a,b\in U^k$, let $a\sim b$ if $b$
		is a permutation of $a$. A \emph{bag} $B$ of cardinality 
		$\card{B}=k$ is a $\sim$-equivalence class on $U^k$.
\end{description}

While the multiplicity perspective better matches the intuitive semantics of
bags, the quotient view later has a closer connection to the probability
spaces we are going to construct.

\subsection{Relational Databases}\label{ssec:db}
We follow the general terminology and notions of the \emph{named perspective}
of databases, see for example \cite{Abiteboul+1995}. We fix two countably
infinite, disjoint sets $\Attributes$ and $\Relations$ of \emph{attribute
names} and \emph{relation names}, respectively. As usual, we drop the
distinction between names of attributes and relations and their model-theoretic
interpretation. A \emph{database schema} is a pair $\Schema = (\Att,\Rel)$ with
the following properties:
\begin{itemize}
	\item $\Att$ and $\Rel$ are finite subsets of $\Attributes$ resp.
		$\Relations$.
	\item For every attribute $A\in\Att$ there exists a set
		$\dom_{\Schema}(A)$, called its \emph{domain}.
	\item For every relation symbol $R\in\Rel$ there exists an associated
		$k$-tuple of distinct attributes from $\Att$ for some $k$,
		called its \emph{type} $\type_{\Schema}(R)$.
\end{itemize}
Implicitly, every relation $R\in\Rel$ has an \emph{arity} $\ar_{\Schema}(R)
\coloneqq \length{\type_{\Schema}(R)}$ and a \emph{domain} $\dom_{\Schema}(R)
\coloneqq \prod_{A\in\type_{\Schema}(R)} \dom_{\Schema}(A)$. Elements of the
domain of $R\in\Rel$ are called \emph{$R$-tuples}. Whenever a pair $(\Att,
\Rel)$ is given, we assume that all of the aforementioned mappings are given as
well, unless it is specified otherwise. Given a database schema $\Schema =
(\Att, \Rel)$ and a relation $R\in\Rel$, the set of \emph{$R$-facts} in
$\Schema$ is formally defined as $\facts_{\Schema} (R) = \set{R} \times
\dom_{\Schema}(R)$. The set of \emph{all} facts of schema $\Schema$ is given as
$\facts_{\Schema}(\Rel)\coloneqq \bigcup_{R\in\Rel}
\facts_{\Schema}(R)$.\par\medskip%

As usual, we denote $R$-facts in the fashion of $R(a_1,\dots,a_k)$ rather than
$(R,a_1,\dots,a_k)$. If $U \subseteq \dom_{\Schema}(R)$ for $R\in\Rel$, we let
$R(U) \coloneqq \set{R(u)\colon u\in U}$. If $U$ is a Cartesian product 
involving singletons, like for example $U = \set{a}\times V$, we may omit the 
braces of the singletons and replace crosses with commas so that $R(a,U) = 
\set{R(a,u)\colon u\in U}$.\par\medskip%

Finally, a \emph{database instance} $D$ of schema $\Schema=(\Att,\Rel)$ is a
\emph{finite} bag of facts from $\facts_{\Schema}(\Rel)$, that is, an element
of the set $\Instances_{\Schema}\coloneqq \Bags{\facts_{\Schema}(\Rel)}
{<\omega}$. We want to emphasize that in particular we allow single facts to 
appear two or more times within an instance. That is, we use bag semantics in 
our database instances.

\subsection{Topology and Measure Theory}\label{ssec:measuretheory}
We assume that the reader is familiar with the basic notions of point set
topology such as open and closed sets and continuous mappings. For a more
detailed introduction to the concepts we use, see \cref{app:topology}. In the
following, we concentrate on the background from measure theory. The
definitions and statements are based upon \cite{Srivastava1998} and Chapter 1
of \cite{Kallenberg1997}.\par
In topological terms, the spaces we use as our attribute domains later on are
called Polish spaces - complete, separable metrizable spaces. Such spaces are
the default choice for probability theory in a general setting, as they are
quite general while still exhibiting the nice behavior of closed intervals of
the real line, in particular the ability to approximate points by converging
sequences of a countable collection 
of open sets.

\begin{example}[{see \cite[ch.~18]{Fristedt+1997} and \cite[pp.~52~et~seqq.]
	{Srivastava1998}}]\label{ex:examplespaces}\mbox{}
	\label{ex:polishspaces}
	\begin{itemize}
		\item All finite and countably infinite spaces (with the
			discrete topology) are Polish.
		\item The spaces $\R$ and $\R\cup\set{\pm\infty}$ are Polish.
		\item Closed subspaces of Polish spaces are Polish.
		\item Countable disjoint unions and countable products of 
			Polish spaces are Polish.
	\end{itemize}
\end{example}

These examples already capture the most relevant cases for standard database 
applications. Nevertheless we stick to the abstract notion of Polish spaces
in order to keep the framework as general as possible. When we work with
Polish spaces, we will later always assume that we work with a fixed metric on
the space (turning it into a complete separable metric space). In particular,
we will use the standard notation $B_\epsilon(x)$ for the \emph{ball of radius
$\epsilon$} around the point $x$ (with respect to said metric).\par\bigskip

Let $\Space[X]$ be some set. A \emph{$\sigma$-algebra} on $\Space[X]$ is a 
family $\Measurable[X]$ of subsets of $\Space[X]$ such that $\Space[X]\in
\Measurable[X]$ and $\Measurable[X]$ is closed under complementation and 
\emph{countable} unions. If $\Measurable[G]$ is a family of subsets of 
$\Space[X]$, then the \emph{$\sigma$-algebra generated by $\Measurable[G]$} is 
the smallest $\sigma$-algebra $\Measurable[X]$ on $\Space[X]$ containing 
$\Measurable[G]$. A \emph{measurable space} is a pair
$(\Space[X],\Measurable[X])$ where $\Space[X]$ is an arbitrary set and
$\Measurable[X]$ is a $\sigma$-algebra on $\Space[X]$. Subsets of $\Space[X]$
are called \emph{$\Measurable[X]$-measurable} (or \emph{measurable} if
$\Measurable[X]$ is clear from context) if they belong to $\Measurable[X]$. A
\emph{probability measure} on $\Space[X]$ is a countably additive function
$P\colon\Measurable[X] \to[0,1]$ with $P(\emptyset)=0$ and $P(\Space[X])=1$.
($P$ being countably additive means $P\big(\bigcup_i \Event[X]_i\big) = \sum_i
P(\Event[X]_i)$ for any sequence $\Event[X]_0,\Event[X]_1,\Event[X]_2,\dots$ of
disjoint measurable sets.) A measurable space equipped with a probability
measure is called a probability space. If $\Xi$ is a probability space
$(\Space[X],\Measurable[X],P)$, we also write $\Pr_{X\sim\Xi} (X\in\Event[X]) =
P(\Event[X])$ or even omit the subscript $X\sim\Xi$, if the underlying probability
space is clear from context.\par\smallskip

Let $(\Space[X],\Measurable[X])$ and $(\Space[Y],\Measurable[Y])$ be measurable
spaces. A mapping $\phi\colon \Space[X]\to \Space[Y]$ is called 
\emph{$(\Measurable[X], \Measurable[Y])$-measurable} (or simply 
\emph{measurable} if the involved $\sigma$-algebras are clear from context) if 
the preimage under $\phi$ of every $\Measurable[Y]$-measurable set is 
$\Measurable[X]$-measurable. That is, if
\begin{equation*}
	\phi^{-1}(\Event[Y]') = \set{X\in\Space[X]\colon\phi(X)\in\Event[Y]'}
	\in\Measurable[X]\qquad\text{for all}~\Event[Y]'\in\Measurable[Y]\text.
\end{equation*}

\begin{fact}[{cf.~\cite[Lemmas 1.4, 1.7 \&{} 1.10]{Kallenberg1997}}]\label{fac:measurablebasics}
	Let $(\Space[X],\Measurable[X])$, $(\Space[Y],\Measurable[Y])$,
	$(\Space[Z],\Measurable[Z])$ be measurable spaces.
	\begin{itemize}
		\item Let $\Measurable[G]$ generate $\Measurable[Y]$. If 
			$\phi\colon\Space[X]\to\Space[Y]$ satisfies $\phi^{-1}
			(\Event[G])\in\Measurable[X]$ for all $\Event[G]\in
			\Measurable[G]$, then $\phi$ is measurable.
		\item If $\phi\colon\Space[X]\to\Space[Y]$ and $\psi\colon
			\Space[Y]\to\Space[Z]$ are measurable, then
			$\psi\circ\phi\colon\Space[X]\to\Space[Z]$ is
			$(\Measurable[X],\Measurable[Z])$-measurable.
		\item If $\Space[Y]$ is a metric space and $(\phi_n)_{n\geq 0}$
			is a sequence of measurable functions $\phi_n\colon
			\Space[X]\to\Space[Y]$ with $\lim_{n\to\infty} \phi_n
			=\phi$, then $\phi$ is measurable as well.
	\end{itemize}
\end{fact}\par

If $(\Space[X],\Topology_{\Space[X]})$ is a topological space, the \emph{Borel
$\sigma$-algebra} $\Borel_{\Space[X]}$ on $\Space[X]$ is the $\sigma$-algebra
generated by $\Topology_{\Space[X]}$. Sets in the Borel $\sigma$-algebra are
also called $\emph{Borel}$.

\begin{fact}[{cf.~\cite[Lemma 1.5]{Kallenberg1997}}]
Any continuous function between the topological 
spaces $(\Space[X],\Topology_{\Space[X]})$ and $(\Space[Y],\Topology
_{\Space[Y]})$ is $(\Borel_{\Space[X]},\Borel_{\Space[Y]})$-measurable 
.
\end{fact}

Two measurable spaces $(\Space[X],\Measurable[X])$ and
$(\Space[Y],\Measurable[Y])$ are called \emph{isomorphic} if there exists a
bijection $\phi\colon\Space[X]\to\Space[Y]$ such that both $\phi$ and
$\phi^{-1}$ are measurable. The mapping $\phi$ is then called an
\emph{isomorphism} between the measurable spaces. If $\Measurable[X]=
\Borel_{\Space[X]}$ and $\Measurable[Y]=\Borel_{\Space[Y]}$, then $\phi$ is
called a \emph{Borel isomorphism} and the measurable spaces are called
\emph{Borel isomorphic}. Measurable spaces that are isomorphic to some Polish
space with its Borel $\sigma$-algebra are called \emph{standard Borel
spaces}.\par\smallskip

If $\Measurable[X]_i$ is a $\sigma$-algebra on $X_i$ for all $i\in I$, the
\emph{product $\sigma$-algebra} $\bigotimes_{i\in I}\Measurable[X]_i$ of
$(\Measurable[X]_i)_{i\in I}$ is the $\sigma$-algebra on $\prod_{i\in
I}\Space[X]_i$ that is generated by the sets $\set{\pi_j^{-1}(\Event[X])\colon
\Event[X]\in \Measurable[X]_j}_{j\in I}$ where $\pi_j$ is the canonical
projection map $\pi_j\colon \prod_{i\in I} \Space[X]_i \to \Space[X]_j$.

\begin{fact}[{cf.~\cite[Lemma 1.2]{Kallenberg1997}}]
	\label{fac:borel-product}%
	Let $(\Space[X]_i)_{i\in I}$ be a \emph{countable} sequence of Polish
	spaces and let $\Borel_i$ be the Borel $\sigma$-algebra of
	$\Space[X]_i$. Then $\Space[X] = \prod_{i\in I} \Space[X]_i$ is Polish
	and $\Borel_{\Space [X]} = \bigotimes_{i\in I} \Borel_i$. That is,
	countable products of standard Borel spaces are standard Borel.
\end{fact}

\subsection{(Finite) Point Processes}\label{ssec:point-processes}
\emph{Point processes} are a well-known concept in probability theory that is
used to model distributions of a discrete (but unknown or even infinite) number
of points in some abstract \enquote{state space}, say the Euclidean space
$\R^n$ \cite{Daley+2003}. They are used to model a variety of both practical
and theoretical problems and appear in a broad field of applications such as,
for example, particle physics, ecology, geostatistics, astronomy and tracking
\cite{Moyal1962,Daley+2003,Degen2015}.
A concrete collection of points that is 
obtained by a draw from such a distribution model is called a
\emph{realization} of the point process. If all realizations are finite, we
speak of a finite point process \cite{Daley+2003}. We proceed to construct a
finite point process over a Polish state space, following the classic
constructions of \cite{Moyal1962, Macchi1975}. While modern point process
theory is much more evolved by casting point processes in the more general
framework of random measures \cite{Daley+2008}, the seminal model of
\cite{Moyal1962,Macchi1975} suffices for our studies due to our restriction to
finite point processes.\par\medskip

Let $(\Space[X],\Measurable[X])$ be a standard Borel space. Then for every $n$,
the product measurable space $(\Space[X]^n,\Measurable[X]^{\otimes n})$ with
$\Measurable[X]^{\otimes n}\coloneqq \Measurable[X]\otimes\dots\otimes\Measurable[X]$
($n$ times) is standard Borel as well (\cref{fac:borel-product}). Letting
$\sim_n$ denote the equivalence relation on $\Space[X]^n$ with
$(x_1,\dots,x_n)\sim_n (y_1,\dots,y_n)$ if there exists a permutation $\pi$ of
$\set{1,\dots,n}$ with $(y_1,\dots,y_n) = (x_{\pi(1)}, \dots,x_{\pi(n)})$, then
elements of $\Space[X]^n/\sim_n$
are basically unordered collections of $n$
(not necessarily different) points, that is, \emph{bags} (or
\emph{multisets}). Formally, we identify $\Space[X]^n/\sim_n$ with the
space $\Bags{\Space[X]}{n}$ of all $n$-element bags from
$X$. The space of all possible realizations is then naturally defined as
\begin{equation*}
	\Bags{\Space[X]}{<\omega}
	= \bigcup_{n\in\N} \Bags{\Space[X]}{n}
	= \bigcup_{n\in\N} \Space[X]^n/\sim_n\text.
\end{equation*}
This is the canonical sample space for a finite point process \cite{Daley+2003,
Moyal1962}, but we need to define a $\sigma$-algebra on this space.  The
original construction of \cite{Moyal1962} considers the
\emph{symmetrization} transformation $\sym$ from $\Space[X]^{<\omega}$ to
$\Bags{\Space[X]}{<\omega}$ where $\sym(x_1,\dots,x_n) = [(x_1,\dots,x_n)]_{\sim_n} =
\bag{x_1,\dots,x_n}$ and $\sym(\Event[X]) = \set{\sym(\bar x)\colon \bar
x\in\Event[X]}$ and defines the $\sigma$-algebra on $\Space[X]$ to be the
set of all subsets of $\Bags{\Space[X]}{<\omega}$ whose preimage under $\sym$
is measurable with respect to the $\sigma$-algebra
on $\Space[X]^{<\omega}$ that is generated using $(\Measurable[X]^{\otimes n})_{n\in\N}$
(pursuing the idea to lift probability measures from well-known product spaces
to the new, in terms of measure theory inconvenient
\enquote{bag-space}---note that the construction above indeed yields a 
$\sigma$-algebra on $\Bags{\Space[X]}{<\omega}$, see \cite[Lemma 1.3]
{Kallenberg1997}). An equivalent, but technically more convenient
construction (see \cite{Macchi1975}) is motivated by an interpretation of
point processes as \enquote{random counting measures}
\cite{Moyal1962,Macchi1975,Daley+2008}: for $\Event[X]\in\Measurable[X]$
and $n\in\N$, the set $\Event[C](\Event[X],n)\subseteq\Bags{\Space[X]}
{<\omega}$ is the set of bags $C$ over $\Space[X]$ with $\hits_C(\Event[X])
\coloneqq \sum_{X\in\Event[X]}\hits_C(X) = n$ (that is, with exactly $n$ 
\enquote{hits} in $\Event[X]$) is called the \emph{counting event} of
$\Event[X]$ and $n$. We define $\Measurable[C]_{\Space[X]}$ to be the
$\sigma$-algebra that is generated by the family of counting events
$\Event[C](\Event[X],n)$ where $\Event[X]$ is Borel in $\Space[X]$ and $n$ is a
nonnegative integer.  The family $\Measurable[C]_{\Space[X]}$ is known as the
\emph{counting $\sigma$-algebra} on $\Bags{\Space[X]}{<\omega}$.
It can be shown that the $\sigma$-algebra generated by the counting
events is the same as the $\sigma$-algebra defined from product
$\sigma$-algebras and the symmetrization operation (see \cite{Moyal1962,Macchi1975}).

\begin{definition}[{cf.~\cite[Def.~1]{Macchi1975}}]
	Let $(\Space[X],\Measurable[X])$ be a standard Borel space and let $P$
	be a probability measure on $\big(\Bags{\Space[X]}{<\omega},
	\Measurable[C]_{\Space[X]}\big)$. Then $\big(\Bags{\Space[X]}{<\omega},
	\Measurable[C]_{\Space[X]},P\big)$ is called a \emph{finite point 
	process} with \emph{state space} $(\Space[X],\Measurable[X])$.\par

	A finite point process $(\Space[Y],\Measurable[Y],P)$ with state space
	$(\Space[X], \Measurable[X])$ is called \emph{simple}, if any 
	realization is almost surely a set, i.\,e. if $\Pr\big(\hits_Y\big(\set{X}\big)\in\set{0,1}~\text{for
all}~X\in\Space[X]\big)= 1$.
\end{definition}

\section{Probabilistic Databases}\label{sec:probabilistic-databases}

In \cite{Grohe+2019}, we introduced a general notion of infinite probabilistic
databases as probability spaces of database instances, that is, probability
spaces $(\Space[D],\Measurable[D],P)$, where $\Space[D]\subseteq
\Instances_{\Schema}$ for some database schema $\Schema$.
Here $\Space[D]$ may be infinite, even uncountable. In fact, in
\cite{Grohe+2019} we only considered instances that are sets rather than bags,
but this does not make much of a difference here.  We left it open, however,
how to construct such probability spaces, and in particular how to define a
suitable measurable spaces $(\Space[D],\Measurable[D])$, which is nontrivial
for uncountable $\Space[D]$. In this section, we provide a general construction
for constructing such measurable spaces.

\subsection{Probabilistic Databases as Finite Point Processes}
\label{ssec:ppconstruction}

\emph{Throughout this paper, we only consider database schemas $\Schema$ where
for every attribute $A$ the domain $\dom_{\Schema}(A)$ is a Polish space.} This
is no real restriction; all domains one might typically find, such as the sets
of integers, reals, or strings over a finite or even countable alphabet have
this property. 

In the following, we fix a database schema $\Schema=(\Att,\Rel)$. It follows
from \cref{fac:borel-product} that not only the domains $\dom_{\Schema}(A)$ of
the attributes $A\in\Att$, but also the spaces $\dom_{\Schema}(R)$ and
$\facts_{\Schema}(R)$ for all $R\in\Rel$ are Polish. We equip all of these
spaces with their respective Borel $\sigma$-algebras and note that
$\dom_{\Schema}(R)$ and $\facts_{\Schema}(R)$ are Borel-isomorphic from the
point of view of measurable spaces.  Thus, they can be used interchangeably
when discussing measurability issues with respect to a single relation. 
For the set $\facts_{\Schema}(R)$ of facts \emph{using relation symbol $R\in
\Rel$}, let $\Measurable[F]_{\Schema}(R)$ denote its (Borel) $\sigma$-algebra. 
We equip $\facts_{\Schema}(\Rel)$, the set of \emph{all facts of schema 
$\Schema$} with the $\sigma$-algebra
\begin{equation*}
	\Measurable[F]_{\Schema}(\Rel)
	= \set{ F \subseteq \facts_{\Schema}(\Rel)\colon 
	F\cap\facts_{\Schema}(R)\in\Measurable[F]_{\Schema}(R)~
	\text{for all}~R\in\Rel}\text.
\end{equation*}
Note that this is indeed a $\sigma$-algebra and, moreover, turns 
$(\facts_{\Schema}(\Rel),\Measurable[F]_{\Schema}(\Rel))$ into a standard Borel 
space (cf. \cite[p.~39]{Fremlin2016} and \cite[p.~166]{Fremlin2013}).

Now a probabilistic database of schema $\Schema$ is supposed to be a
probability space $(\Space[D],\Measurable[D],P)$ where $\Space[D]\subseteq
\Instances_{\Schema}$. Without loss of generality we may assume that actually
$\Space[D]=\Instances_{\Schema}= \Bags{\facts_{\Schema}(\Rel)}{<\omega}$,
because we can adjust the probability measure to be $0$ on instances we are not
interested in. Thus a probabilistic database is a probability space over finite
sets of facts. This is exactly what a finite point process over the state space
consisting of facts is. We still need to define the $\sigma$-algebra
$\Measurable[D]$, but the theory of point processes gives us a generic way of
doing this: we let $\Measurable[D]_{\Schema} = \Measurable[C]_{\facts_{\Schema}
(\Rel)}$ be the counting $\sigma$-algebra of $\Instances_{\Schema}$ (cf. 
\cref{ssec:point-processes}).

\begin{definition}\label{def:pdb}
  A \emph{standard probabilistic database} of schema $\Schema$ is a
  probability space $(\Instances_{\Schema}, \Measurable[D]_{\Schema},P)$.
\end{definition}

That is, a standard probabilistic database of schema $\Schema$ is a finite
point process over the state space $(\facts_{\Schema}(\Rel),
\Measurable[F]_{\Schema})$. 

The reason we speak of \enquote{standard} PDBs in the definition is to
distinguish them from the more general PDBs introduced in \cite[Definition 3.1]
{Grohe+2019}. In \cite{Grohe+2019}, we left the $\sigma$-algebra unspecified
and only required the (mild) property, that the occurrence of measurable sets
of facts is themselves measurable. This requirement corresponds to a set
version of the counting events defined above and is thus given by default in a
standard probabilistic database.\par

Even though the construction of counting $\sigma$-algebras for point processes
is nontrivial, we are convinced that it is a natural generic construction of
$\sigma$-algebras over spaces of finite (or countable) sets and the extensive
usage of these constructions throughout mathematics for more than fifty years
now indicates their suitability for such tasks. \emph{Throughout this paper,
all probabilistic databases are standard. Therefore, we omit the qualifier
\enquote{standard} in the following and just speak of probabilistic databases
(PDBs).}\par

We defined instances of PDBs to be bags of facts. However, if a PDB, that
is, a finite point process is simple (see \cref{ssec:point-processes}), then it
may be interpreted as a PDB with set-instances.

\begin{example}\label{ex:finitepdb}
	Every finite probabilistic database (as introduced, for example,
	in $\cite{Suciu+2011}$) can be viewed as a standard PDB: Let
	$\tilde{\Space[D]}$ be a finite set of set-valued database instances
	over some schema $\Schema=(\Att,\Rel)$ and let $\tilde
	P\colon\tilde{\Space[D]} \to[0,1]$ a probability measure on
	$\tilde{\Space[D]}$ (equipped with the power set as its
	$\sigma$-algebra). Then $(\tilde{\Space[D]}, \tilde P)$ corresponds to
	the simple finite point process $(\Space[D],\Measurable[D],P)$ on the
	instance measurable space of $\Schema$ with state space
	$(\facts_{\Schema}(\Rel),\Measurable[F]_{\Schema}(\Rel))$ where
	$P(\Event[D])=\tilde P(\Event[D]\cap \tilde{\Space[D]})$ (interpreting
	$\tilde{\Space[D]}$ with a (finite) collection of bags with
	$\set{0,1}$-valued multiplicities).  
\end{example}

\subsection{The Possible Worlds Semantics of Queries and Views}\label{ssec:pws}
In the traditional database setting, \emph{views} are mappings from database
instances of an \emph{input schema} (or \emph{source schema}) $\Schema = 
(\Att,\Rel)$ to database instances of some \emph{output schema} (or
\emph{target schema}) $\Schema' = (\Att',\Rel')$. Views, whose output schema
$\Schema'$ consists of a single relational symbol only are called
\emph{queries}. Queries and views are usually given by syntactic
expressions in some \emph{query language}. As it is common, we will blur the
distinction between a query (or view) and its syntactic
representation.

Let $\pdb = (\Instances_{\Schema}, \Measurable_{\Schema}, P)$ be a
probabilistic database of schema $\Schema = (\Att, \Rel)$ and let $V$ be a view
of input schema $\Schema$ and output schema $\Schema' = (\Att', \Rel')$. The
image of a set $\Event \subseteq \Measurable$  of instances is $V(\Event) =
\set{ V(D) \colon D\in\Event }\subseteq \Instances_{\Schema'}$.

Now we would like to define a probability measure on the output space
$(\Instances_{\Schema'}, \Measurable_{\Schema'})$ by 
\begin{equation}\label{eq:pwsprobability}
	P'(\Event')
	\coloneqq P\big( V^{-1}(\Event') \big)
	= P\big( \set{ D\in\Instances \colon V(D)\in\Event' } \big)
\end{equation}
for all $\Event'\in\Measurable_{\Schema'}$. Then $V$ would map
$\pdb$ to $\pdb'\coloneqq (\Instances_{\Schema'},
\Measurable_{\Schema'},P')$. This semantics of views over PDBs is 
known as the \emph{possible worlds semantics} of probabilistic databases
\cite{Green2009,Aggarwal+2009,Suciu+2011,VandenBroeck+2017}.

However, $P'$ (as defined in \eqref{eq:pwsprobability}) is only well-defined if for all $\Event'\in\Measurable_{\Schema'}$ the set
$V^{-1}(\Event')$ is in $\Measurable_{\Schema}$, that is, if $V$ is a
\emph{measurable} mapping from $(\Instances_{\Schema},
\Measurable_{\Schema})$ to $(\Instances_{\Schema'},
\Measurable_{\Schema'})$. 

Measurability is not just a formality, but an issues that requires
attention. The following 
example shows that there are relatively simple \enquote{queries} that are
not measurable.  

\begin{example}\label{ex:nonborel}
  Let $\Schema=\Schema'$ be the schema consisting of a singe unary relation
  symbol $R$ with attribute domian $\R$ (equipped with the Borel $\sigma$-algebra), and let $B$ be some Borel
	set in $\R^2$. 

  We define a mapping $Q_B\colon\Instances_{\Schema}\to\Instances_{\Schema}$, our
  \enquote{query}, by     
  \begin{equation*}
    Q_B(D)\coloneqq
    \begin{dcases}
      D &\text{if}~D~\text{is a singleton}~\bag{R(x)}~\text{and
        there exists }~y\in\R~\text{s.\,t.}~(x,y)\in B,\\
      \emptyset &\text{otherwise.}
    \end{dcases}
  \end{equation*}
  Observe that 
  $
    Q_B^{-1}(\Instances_{\Schema})=\set[big]{\bag{R(x)}\colon x\in
      \proj_1(B)},
    $
  where
  $\proj_1(B)=\set{x\in\R \colon \text{there is}~y\in\R~\text{s.\,t.}~(x,y)\in
    B}$. It is a well known fact that there are Borel sets $B\subseteq\R^2$
    such that the projection $\proj_1(B)$ is \emph{not} a Borel set in $\R$ 
	 (see \cite[Theorem 4.1.5]{Srivastava1998}). For such sets $B$, the query 
 $Q_B$ is not measurable.\end{example}

The rest of this paper is devoted to proving that queries and views
expressed in standard query languages, specifically relational algebra,
possibly extended by aggregation, and datalog queries, are measurable
mappings and thus have a well-defined open-world semantics over
probabilistic databases.

It will be sufficient to focus on 
\emph{queries}, because views can be composed from queries and the
measurability results can be lifted (as we formally show in the next
subsection). \emph{Throughout the rest of the paper, we adopt the following
notational conventions: queries are denoted by $Q$
	and map a PDB $\pdb=(\Instances,\Measurable,
	P)$ to a PDB $\pdb'=(\Instances',\Measurable',P')$ such that $\pdb$ is of
schema $\Schema$ and $\pdb'$ is of schema $\Schema'$.}

\begin{observation}\label{obs:simplifying}
	The task of establishing measurability of queries in our framework is
	simplified by the following.
	\begin{enumerate}
		\item If we want to demonstrate the measurability of $Q$, it
			suffices to show that $Q^{-1}(\Event') \in \Measurable$
			for all \emph{counting events} $\Event[D]'=\Event[C]'
			(F,n)$ of $(\Instances',\Measurable')$. This is due to 
			\cref{fac:measurablebasics} because they generate 
			$\Measurable'$.
		\item Since compositions of measurable mappings are measurable
			(again from \cref{fac:measurablebasics}), composite queries are
			immediately measurable if all their components are
			measurable queries to begin with. In particular, we
			can demonstrate the measurability of general queries of
			some query language by structural induction.
	\end{enumerate}
\end{observation}

\begin{remark}
Let us again mention something related to the well-established knowledge on
point processes. The mappings (queries) we investigate map between point 
processes that are defined on different measure spaces that are themselves a
conglomerate of simpler measure spaces of different shape. It is well-known
that measurable transformations of the state space of a point process define
a new point process on the transformed state space (a strengthening of this
result is commonly referred to as the \enquote{mapping theorem} 
\cite{Last+2017}). Our queries however are in general already defined on point
configurations and not on the state space of facts. Thus, their measurability
can in general not be obtained by the idea just sketched.
\end{remark}

\subsection{Assembling Views from Queries}
We think of views as finite sets of queries, including one for every 
relation of the output schema. Suppose $V = \set{ Q_1, \dots, Q_k }$ is a view
consisting of measurable queries $Q_1, \dots, Q_k$ where the names of the
target relations of the $Q_i$ are mutually distinct. The target schema 
$\Schema'$ of $V$ is given by the union of the target schemas of $V$s 
individual queries. Now every fact $f \in \facts_{\Schema'}(\Rel')$ of the new
schema originates from the target schema of exactly one of the queries $Q_1,
\dots, Q_k$. We refer to that query as $Q_f$. Then for all $D \in \Instances$
and $f \in \facts_{\Schema'}(\Rel')$, we define $\hits_{V(D)}(f) \coloneqq 
\hits_{Q_f(D)}(f)$.  Now if $F\subseteq\facts_{\Schema'}(\Rel')$, let $F_i \coloneqq F \cap
\facts_{\Schema_i'}(\Rel_i')$ where $\Schema_i' = (\Att_i', \Rel_i')$ is the
target schema of $Q_i$. Then
\begin{equation*}
	\hits_{V(D)}(F) = n
	\Iff \text{there are}~n_1,\dots,n_k~\text{with}~\textstyle\sum_{i=1}^k 
		n_i=n~\text{such that}~\hits_{Q_i(D)}(F_i)=n_i\text.
\end{equation*}
Since the $F_i$ are measurable if and only if $F$ is measurable, the above
describes a countable union of measurable sets. Thus, $V$ is
measurable.\par\smallskip

\section{Relational Algebra}\label{sec:relational-algebra}
As motivated in \cref{ssec:pws}, we now investigate the measurability of
relational algebra queries in our model. The concrete relational algebra for
bags that we use here is basically the (unnested version of the) algebra that 
was introduced in \cite{Dayal+1982} and investigated respectively extended and 
surveyed in \cite{Albert1991,Grumbach+1996a,Grumbach+1996b}. It is called 
$\mathsf{BALG^1}$ (with superscript $1$) in \cite{Grumbach+1996a}. We do not
introduce nesting as it would yield yet another layer of abstraction and
complexity to the spaces we investigate, although by the properties that such
spaces exhibit, we have strong reason to believe that there is no technical
obstruction in allowing spaces of finite bags as attribute domains.  

The operations we consider are shown in the \cref{tab:balg} below.
As seen in \cite{Albert1991,Grumbach+1996a,Grumbach+1996b}, there is some
redundancy within this set of operations that will be addressed later. A
particular motivation for choosing this particular algebra is that possible
worlds semantics are usually built on top of set semantics and these operations
naturally extend the common behavior of relation algebra queries to bags. This
is quite similar to the original motivation of \cite{Dayal+1982} and 
\cite{Albert1991} regarding their choice of operations. A detailed overview
of their traditional semantics on single database instances can be found in 
\cref{app:operations}.
\begin{table}[H]
	\centering
	\caption{$\mathsf{BALG^1}$-operators considered in this paper (see 
	\cref{app:operations} for details).}\label{tab:balg}
	\begin{tabular}{lll}
		\toprule
		\textbf{Base Queries}
		& Constructors  & $Q=\bag{}$ and $Q=\bag{R(a)}$\\
		& Extractors 	& $Q=R$\\
		& Renaming	& $Q=\rename[A\to B](R)$\\
		\midrule
		\textbf{Basic Bag Operations} 
		& Additive Union 	& $Q=R_1\addunion R_2$\\
		& Difference 		& $Q=R_1\difference R_2$\\
		& Max-Union		& $Q=R_1\maxunion R_2$\\
		& (Min-)Intersection	& $Q=R_1\intersection R_2$\\
		& Deduplication		& $Q=\dedupe(R)$\\
		\midrule
		\textbf{SPJ-Operations}
		& Selection		& $Q=\select[{(A_1,\dots,A_k)\in\Event[B]}](R)$\\
		& Projection		& $Q=\project[(A_1,\dots,A_k)](R)$\\
		& Cross Product		& $Q=R_1\product R_2$\\
		\bottomrule
	\end{tabular}
\end{table}
 The main result we establish in this section is the 
following theorem:
\begin{theorem}\label{thm:balg}
	All queries expressible in the bag algebra $\mathsf{BALG^1}$ are measurable.
\end{theorem}
Since compositions of measurable mappings are measurable, the measurability of
the operators from \cref{tab:balg} directly entails the measurability of
compound queries by structural induction.\par
First note that the measurability of the base queries is easy to prove.
\begin{lemma}
	The queries $\bag{}$, $\bag{R(a)}$ and $R$ are measurable.
\end{lemma}
\begin{proof}
	First consider $Q = \bag{}$ and fix some $\Event[D]'\in
	\Measurable[D]'$. If $\bag{}\in\Event[D]'$, then
	$Q^{-1}(\Event[D]')=\Instances\in\Measurable$. Otherwise,
	$Q^{-1}(\Event[D]')=\emptyset\in\Measurable$. Thus, $Q$ is measurable.
	The same argument applies to $Q= \bag{R(a)}$.\par
	Now consider the query $Q = R$ and let $\Event[C]'(F,n)$ be a counting
	event in the output measurable space. Then for every instance
	$D\in\Instances$, $\hits_{Q(D)}(F) = n$ if and only if $\hits_D(F)=n$
	Thus, $Q^{-1}(\Event[C]'(F,n))$ is the counting event $\Event[C](F,n)$
	in $(\Instances,\Measurable)$.  Hence, $Q$ is measurable.\qedhere
\end{proof}

\subsection{Basic Bag Operations}
We will obtain the measurability of the basic bag operations $\addunion$,
$\difference$, $\intersection$, $\maxunion$, $\dedupe$ as a consequence 
of the following, more general result that gives some additional insight into 
properties that make queries measurable.

Consider a query $Q$ of input schema $\Schema$ and output schema $\Schema'$
operating on relations $R_1$ and $R_2$ of $\Schema$. Let $R'$ be the single
(output) relation of $\Schema'$.

\begin{lemma}\label{lem:qcrit}
	Suppose that given $Q$ there exist functions $q_1 \colon 
	\facts_{\Schema'}(R') \to \facts_{\Schema}(R_1)$ and $q_2\colon 
	\facts_{\Schema'}(R')\to\facts_{\Schema}(R_2)$ with the following
	properties:
	\begin{enumerate}
	\item for all $n\in\N$ there exists a set $M(n)\subseteq\N^2$ with $(0,0)
          \notin M(n)$ for $n>0$ such that for all $D\in\Instances$ and 
		all $f\in\facts_{\Schema'}(R')$ it holds that
		\begin{equation*}
			\hits_{Q(D)}(f)=n\qquad
			\text{if and only if}\qquad
			\big(\hits_D(q_1(f)),\hits_D(q_2(f))\big)\in M(n)\text;
		\end{equation*}
	\item both $q_1$ and $q_2$ are injective and continuous;
	\item the images of $F$ under $q_1$ and $q_2$ are measurable: 
		$q_1(F)\in\Measurable[F]_{\Schema}(R_1)$ and $q_2(F)\in
		\Measurable[F]_{\Schema}(R_2)$.
	\end{enumerate}
	Then $Q$ is measurable.
\end{lemma}

Let us briefly mention the impact of the various preconditions of the lemma 
before turning to its proof. The existence of the functions $q_1$ and $q_2$ 
ensures that preimages of counting events $\Event[C]'(F,n)$ under the
query $Q$ can be approximated by using the fact that our state spaces are
Polish. They \enquote{decompose} the set $F$ of facts into disjoint (and 
measurable!) sets of facts for the preimage in a continuous, invertible way 
that exactly captures how tuples in the preimage relate to tuples in the image.

\begin{proof}[{\proofname~(\cref{lem:qcrit})}]
	Assume that $q_1$ and $q_2$ exist with properties 1 to 3. We fix $F\in
	\Measurable[F]_{\Schema'}(R')$ and $n\in\N_+$ and show that $Q^{-1}
	(\CEvent'(F,n))$ is in $\Measurable$. Let $F_0$ be a countable, dense
	set in $\facts_{\Schema'}(R')$. We claim that $\hits_{Q(D)}(F)=n$ if 
	and only if
	\begin{equation*}
		\tag{$*$}\label{eq:bigone}%
		\begin{aligned}
			&\text{there exist}~\ell\in\N_+~\text{and}~n_1,\dots,
			n_{\ell}\in\N~\text{with}~\textstyle\sum_{i=1}^\ell 
			n_i = n~\text{and}\\
			&\quad[1]\text{there exist}~(n_{i,1},n_{i,2})\in M(n_i)
			~\text{and}~k_0\in\N_+~\text{and}\\
			&\quad[2]\text{there exist Cauchy sequences}~(f_1^k)_
			{k\in\N},\dots,(f_\ell^k)_{k\in\N}~\text{in}~F_0~
			\text{with}\\
			&\quad[2] B_{1/k_0}(f_i^k)\cap B_{1/k_0}
				(f_{i'}^{k'}) = \emptyset~\text{for
                                  all}~k,k'\text{ and }i\neq i'~
				\text{such that for all}~k>k_0\\
			&\quad[3] \hits_{D}(q_1(F)\cap B_{1/k}(q_1(f_i^k)))=n_{i,1}~\text{and}~
			\hits_{D}(q_2(F)\cap B_{1/k}(q_2(f_i^k)))=n_{i,2}\\
			&\quad[4]\text{for}~1\leq i\leq \ell~\text{and}\\
			&\quad[3] \hits_{D}(q_1(F)\setminus\textstyle\bigcup_{i=1}^\ell B_{1/k}(q_1(f_i^k)))=0~\text{and}~
			\hits_{D}(q_2(F)\setminus\textstyle\bigcup_{i=1}^\ell B_{1/k}(q_2(f_i^k)))=0\text.
		\end{aligned}
	\end{equation*}
	Note that \eqref{eq:bigone} is a countable combination of counting
	events in $(\Instances,\Measurable)$ (using condition 3, in 
	particular). Thus, to show the measurability of $Q$ it suffices to 
	show the equivalence of $\hits_{Q(D)}(F)=n$ and	\eqref{eq:bigone}.
	\par\medskip%
	\begin{figure}[H]
		\centering
		\caption{Example illustration of \eqref{eq:bigone} for two
			facts $f$ and $f'$. Both these facts are approximated
			by Cauchy sequences that under $q_1$ and $q_2$ also
			approximate their images.}
		\begin{tikzpicture}
			\tikzset{%
				point/.style={inner sep=0pt,circle,fill,minimum width=4pt}
			}
			\draw (-5,1) rectangle (-3,3.3);
			\node at (-4,.5) {$\facts_{\Schema'}(R')$};
			\draw[pattern=north west lines,pattern color=black!10!white,smooth cycle] plot  coordinates {(-4.4,1.2) (-3.6,2) (-3.5,3.0) (-4.8,3.1)};
			\node at (-3.6,1.5) {$F$};
			\node[point,label={[shift={(.1,0)}]left:$f$}] (f) at (-4.2,2.3) {};
			\node[point,gray,label={[shift={(.2,-.2)}]above left:$f^k$}] (fk) at (-4.25,2.6) {};
			\draw[gray] (fk) circle (.4cm);
			\draw[-stealth,dashed,gray] (fk) to (f);
			\node[point,label={[shift={(.1,.1)}]below left:$f'$}] (f') at (-4,2) {};

			\draw[xshift=5.5cm,pattern=north west lines,pattern 
				color=black!10!white,smooth cycle] plot 
				coordinates {(-4.4,1.2) (-3.2,2.2) (-4.8,2.8)};
			\node[xshift=5.5cm] at (-3.4,1.4) {$\scriptstyle q_1(F)$};

			\draw[xshift=5.5cm] (-5,1) rectangle (-3,3.3);
			\node[xshift=5.5cm] at (-4,.5) {$\facts_{\Schema}(R_1)$};
			
			\node[xshift=5.5cm,point,gray] (a) at (-3.6,2.6) {};
			\node[xshift=5.5cm,point] (a') at (-3.65,2.3) {};
			\draw[gray] (a) circle (.4cm);
			\node[xshift=5.5cm,point,gray] (b) at (-4.5,1.8) {};
			\node[xshift=5.5cm,point] (b') at (-4.25,1.85) {};
			\draw[gray] (b) circle (.4cm);

			\draw[xshift=8.5cm,pattern=north west lines,pattern 
				color=black!10!white,smooth cycle] plot 
				coordinates {(-4.4,1.9) (-3.6,1.6) (-3.2,3) (-4.8,3.2)};
			\node[xshift=8.5cm] at (-3.4,1.4) {$\scriptstyle q_2(F)$};

			\draw[xshift=8.5cm] (-5,1) rectangle (-3,3.3);
			\node[xshift=8.5cm] at (-4,.5) {$\facts_{\Schema}(R_2)$};

			\node[xshift=8.5cm,point,gray] (c) at (-3.7,2.8) {};
			\node[xshift=8.5cm,point] (c') at (-3.8,2.6) {};
			\draw[gray] (c) circle (.4cm);
			\node[xshift=8.5cm,point,gray] (d) at (-4.2,1.6) {};
			\node[xshift=8.5cm,point] (d') at (-4,1.8) {};
			\draw[gray] (d) circle (.4cm);

			\draw[-stealth,shorten >=0.05cm,bend left=10,very thick] (f) to node[midway,below,xshift=.8cm,yshift=-0.1cm]{$q_1$}(a');
			\draw[-stealth,shorten >=0.05cm,bend left=10,very thick] (f) to node[midway,below,xshift=2.8cm,yshift=-0.1cm]{$q_2$} (c');
			\draw[-stealth,dashed,gray,bend left=10,shorten >=0.05cm,very thick] (fk) to (a);
			\draw[-stealth,dashed,gray,bend left=10,shorten >=0.05cm,very thick] (fk) to (c);
		\end{tikzpicture}
	\end{figure}
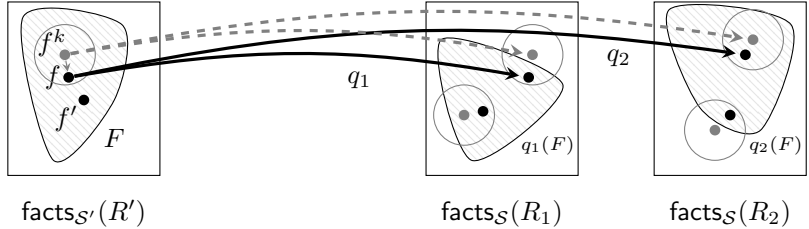

	Assume $\hits_{Q(D)}(F)=n$. Let $f_1,\dots,f_\ell$ be the facts from 
	$F$ with the property that $\hits_D(q_1(f))>0$ or
        $\hits_D(q_2(f))>0$.

	Let $n_i \coloneqq \hits_{Q(D)}(f_i)$. From condition 1 we know that
	$(\hits_D(q_1(f_i)),\hits_D(q_2(f_i)))\in M(n_i)$ as well as $\sum_{i=1}
	^\ell n_i=n$. Let $(f_1^k),\dots,(f_\ell^k)$ be Cauchy sequences from $F_0$
	that converge to $f_1,\dots,f_\ell$. Since $\ell$ is finite, the balls
	around $f_i^k$ and $f_{i'}^k$ do not intersect for sufficiently large $k$ as
	well as the balls around their images under $q_1$ respectively $q_2$ (since
	both of them are injective and continuous). Thus, $\hits_D(q_1(F)\cap 
	B_{1/k}(q_1(f_i^k))) = \hits_D(q_1(f_i))$ and $\hits_D(q_2(F)\cap B_{1/k}
	(q_2(f_i^k))) = \hits_D(q_2(f_i))$ for sufficiently large $k$. Therefore, 
	$D$ satisfies \eqref{eq:bigone}.\par\medskip

	Now for the other direction, suppose $D$ satisfies \eqref{eq:bigone}.
	As the $f_i^k$ are Cauchy sequences, the spaces $\facts_{\Schema'}(R_j)$ are
	Polish and hence complete, and the $q_j$ are continuous
	there exists (for every $1\leq i\leq \ell$) some $f_i\in F$ such that 
	$f_i^k\to f_i$, $q_1(f_i^k)\to q_1(f_i)$ and $q_2(f_i^k)\to q_2(f_i)$ 
	as $k\to\infty$ and $(\hits_D(q_1(f_i)),\hits_D(q_2(f_i)))=
	(n_{i,1},n_{i,2})\in M(n_i)$. By condition 1, $Q(D)$ contains $f_i$
	with multiplicity $n_i$ and as $\sum_{i=1}^\ell n_i=n$ (and since
	$D$ had no other facts with positive multiplicity than the above),
	it follows that $\hits_{Q(D)}(F)=n$.
\end{proof}

Note that the result above easily generalizes to queries that depend on an
arbitrary number of relations of the input probabilistic database. \cref{lem:qcrit}
provides a criterion to establish the measurability of queries. Checking its
precondition for bag operations we consider turns out to be quite easy and
yields the following lemma.

\begin{lemma}
	The following queries are measurable:
	\begin{enumerate}
		\item\label{itm:addunion} \emph{(Additive Union)} $Q = R_1\addunion R_2$ with $R_1,R_2\in
			\Rel$ of equal type.	
		\item\label{itm:difference} \emph{(Difference)} $Q = R_1\difference R_2$ with $R_1,R_2
			\in\Rel$ of equal type.
		\item\label{itm:intersection} \emph{((Min-)Intersection)} $Q = R_1\intersection R_2$ with $R_1,
			R_2\in\Rel$ of equal type.
		\item\label{itm:maxunion} \emph{(Max-Union)} $Q = R_1\maxunion R_2$ with $R_1,R_2\in
			\Rel$ of equal type.
		\item\label{itm:dedupe} \emph{(Deduplication)} $Q = \dedupe(R)$ with $R\in\Rel$.
	\end{enumerate}
\end{lemma}

\begin{proof}\crefname{enumi}{Statement}{Statements}
	As $\maxunion$ and $\intersection$ are expressible via 
	$\addunion$ and $\difference$ (cf. \cite{Albert1991}), we only
	show \cref{itm:addunion,itm:difference,itm:dedupe}.

	\begin{enumerate}
		\item Define $q_1$ and $q_2$ by $q_i(R(x))=R_i(x)$. Then
			$q_i$, $i\in\set{1,2}$ is injective and continuous
			and $q_i(F)=F_i\in\Measurable[F]_{\Schema}(R_i)$. Now
			let $k\in\N$ and let $M(k)\subseteq\N^2$ be the set of
			pairs $(k_1,k_2)$ with the property that $k_1+k_2=k$.
			Then $\hits_{Q(D)}(f) = k$ if and only if 
			$\big(\hits_D\big(q_1(f)\big),\hits_D\big(q_2(f)\big)\big)\in M(k)$. Together, by
			\cref{lem:qcrit}, $Q$ is measurable.
		\item This works exactly like in the case of $\addunion$ with
			$M(k)$ being the set of pairs $(k_1,k_2)$ with $\max(k_1-k_2,0)=k$.
			\setcounter{enumi}{4}
		\item In this case, we only use a single function $q$ that
			maps $R'(x)$ to $R(x)$. Again, $q$ is obviously both
			continuous and injective and $q(F)\in
			\Measurable[F]_{\Schema}(R)$ for every measurable $F$.
			If $k=1$, we let $M(k)=\N\setminus\set{0}$ and $M(k)=\set{0}$
			otherwise. Then clearly $\hits_{Q(D)}(R(x))=k$ if and only
			if $\hits_D(q(R(x))\in M(k)$ and again, $Q$ is measurable
			by \cref{lem:qcrit}.\qedhere
	\end{enumerate}
\end{proof}

\subsection{Selection, Projection and Join}
In this section, we investigate selection and projection as well as the cross
product of two relations. We start with the following helpful lemma that allows
us to restructure our relations into a more convenient shape to work with.
Semantically, it might be seen as a special case of a projection query.

\begin{lemma}\label{lem:reordering}%
	Reordering attributes within the type of a relation yields a
	measurable query.
\end{lemma}

\begin{proof}
	Recall that any permutation can be expressed as a composition of 
	transpositions. Thus, we only consider the case where two attributes,
	say $A$ and $B$, switch places within the type of some relation $R\in
	\Rel$. Let $q$ be the function that maps $\facts_{\Schema}(R)$ to 
	$\facts_{\Schema'}(R')$ by swapping the entries for attribute
	$A$ and $B$. Obviously, under $q$, the preimage of a measurable 
	rectangle in $\Measurable[F]_{\Schema'}(R)$ is a measurable rectangle
	itself. As $\hits_{Q(D)}(F)=n$ if and only if $\hits_{D}(q^{-1}(F))=n$,
	$Q$ is measurable.
\end{proof}

\begin{lemma}\label{lem:selection-measurable}
	The query $Q = \select[{(A_1,\dots,A_k)\in\Event[B]}](R)$ is measurable
	for all $R\in\Rel$, all pairwise distinct attributes $A_1,\dots,A_k
	\in\type_{\Schema}(R)$ and all Borel subsets $\Event[B]$ of 
	$\prod_{i=1}^k \dom_{\Schema}(A_i)$.
\end{lemma}

\begin{proof} Fix some $F\in\Measurable[F]_{\Schema'}(R')$ and $n\in\N$. By 
	\cref{lem:reordering}, we may assume that $\type_{\Schema}(R) = 
	(A_1,\dots,A_m)$ where $m\geq k$. Let $F_{\Event[B]}\coloneqq \set{R}
	\times\Event[B]\times\dom_{\Schema}(A_{k+1})\times\cdots\times
	\dom_{\Schema}(A_m)$. Note that $F_{\Event[B]}\in
	\Measurable[F]_{\Schema}(R)$. (This is a consequence of 
	\cref{fac:borel-product}.) As $\hits_{Q(D)}(F)=n$ if and only if
	$n=\hits_{Q(D)}(F\cap F_{\Event[B]}) = \hits_D(F\cap F_{\Event[B]})$, $Q$
	is measurable.
\end{proof}

\begin{example}
	Assume that $\dom_{\Schema}(A) = \dom_{\Schema}(B) = \R$ and both $A$
	and $B$ appear in the type of $R\in\Rel$. It is well-known (and can
	be shown by standard arguments) that the sets $\Event[B]_{=}\coloneqq
	\set{(x,y)\in\R^2\colon x=y}$ and $\Event[B]_{<}\coloneqq\set{(x,y)\in
	\R^2\colon x<y}$ are Borel in $\R^2$. Thus $\select[A=B](R)\coloneqq
	\select[{(A,B)\in\Event[B]_{=}}](R)$ and $\select[A<B](R)\coloneqq
	\select[{(A,B)\in\Event[B]_{<}}](R)$ are measurable by 
	\cref{lem:selection-measurable}.
\end{example}

\begin{lemma}\label{lem:projection-measurable}.
	The query $Q = \project[A_1,\dots,A_k](R)$ is measurable for all $R\in
	\Rel$ and all mutual distinct $A_1,\dots,A_k\in\type_{\Schema}(R)$.
\end{lemma}

\begin{proof} Again, fix some $F\in\Measurable[F]_{\Schema'}(R')$ and $n\in\N$.
	Note that $F$ is of the shape $\set{R'}\times\Event[B]$ where 
	$\Event[B]$ is Borel in $\dom_{\Schema'}(R')=\prod_{i=1}^k 
	\dom_{\Schema}(A_k)$. By \cref{lem:reordering}, we may again assume
	that $\type_{\Schema}(R) = (A_1,\dots,A_m)$ with $m\geq k$. Define
	$F_{\Event[B]}$ exactly like in the proof of 
	\cref{lem:selection-measurable}: $F_{\Event[B]}\coloneqq\set{R}\times
	\Event[B]\times\dom_{\Schema}(A_{k+1})\times\cdots\times 
	\dom_{\Schema}(A_m)$.  Again, $F_{\Event[B]}\in
	\Measurable[F]_{\Schema}(R)$. Now, we have $\hits_{Q(D)}(F)=n$ if and
	only if $\hits_D(F_{\Event[B]})=n$ and hence, $Q$ is measurable.
\end{proof}

\begin{lemma}\label{lem:crossprod-measurable}
	The query $Q = R_1 \product R_2$ is measurable for all $R_1, R_2 \in
	\Rel$.
\end{lemma}

First we note that this turns out to be more involved than it seems on first
sight. The straight-forward approach would be to take a counting event
$\CEvent(F,n)$ in the output measurable space and to decompose $F$ into its
\enquote{left and right parts} $F_1\subseteq\facts_{\Schema}(R_1)$ and $F_2
\subseteq\facts_{\Schema}(R_2)$ such that the instances from the preimage of
the query are exactly those with $\hits_D(F_1)=n_1$ and $\hits_D(F_2)=n_2$ such
that $n_1\cdot n_2=n$, similar to the setting of \cref{lem:qcrit}. This 
approach does not settle the case since the sets $F_1$ and $F_2$ need not be
measurable in general (see \cite[Theorem 4.1.5]{Srivastava1998}; we used the 
same argument in \cref{ex:nonborel}) which in particular violates the second 
precondition of \cref{lem:qcrit}.

\begin{sketch}
	Using renaming, we may assume that the types of $R_1$ and $R_2$ are
	disjoint in terms of attribute names. Consider $F\in\Measurable[F]_
	{\Schema'}(R')$ and $n\in\N$. If $F$ is a measurable rectangle
	$F = F_1\times F_2$, it is easy to see that the na\"ive approach
	sketched above works via $\hits_{Q(D)}(F)=n$ if and only if
	$\hits_D(F_1)\cdot\hits_D(F_2)= n$.\par
	In the general case of $F$ being an arbitrary Borel set, we consider
	the \emph{$k$-coarse} preimage of $\CEvent'(F,n)$ first. These are the
	database instances from $\Instances$ whose minimal inter-tuple
	distance is at least $\frac{1}{k}$ for some fixed Polish metrics. One
	can show that these $k$-coarse preimages of the query are measurable
	for all $F,n$ and $k$. As the union of these preimages over all
	positive integers $k$ is exactly the preimage of $\CEvent'(F,n)$, $Q$
	is measurable. The details of this proof are shown in
	\cref{app:crossprod}.
\end{sketch} \par

Altogether, within the last three sections, we have established the measurability of all
the (bag) relational algebra operators from \cref{tab:balg} and thus have
proven \cref{thm:balg}. Of course any additional operator that is expressible 
by a combination of operations from \cref{tab:balg} is immediately measurable
as well, including for example \emph{natural joins} $Q=R_1\natjoin R_2$ or
selections where the selection predicate is a Boolean combinations of
predicates of the shape $(A_1,\dots,A_k)\in\Event[B]$.

\section{Aggregate Queries}\label{sec:aggregate-queries}
In this section, we study various kinds of aggregate operators. Let $U$ and $V$
be standard Borel spaces. An \emph{aggregate operator} (or \emph{aggregator})
from $U$ to $V$ is a mapping $\Phi$ that sends bags of elements of $U$ to
elements of $V$: $\Phi\colon\Bags{U}{<\omega} \to V$. Every such aggregator
$\Phi$ gives rise to a query $Q=\varpi_\Phi(R)$ defined by $Q(D)\coloneqq
\bag{R'(v)}$ for $v\coloneqq\Phi(\bag{ u \colon R(u)\in D })$. (The 
notation we use for aggregation queries is loosely based on that of 
\cite{Fink+2012}.) Observe that for every instance $D$, $\hits_{Q(D)}\big(R'(v)
\big) = 1$ if and only if $\Phi(\bag{ u \colon R(u)\in D })=v$ (and $0$ 
otherwise). It is easy to see that $Q=\varpi_\Phi(R)$ is a measurable query
whenever $\Phi$ is measurable w.\,r.\,t.\ the counting $\sigma$-algebra on $\Bags{U}
{<\omega}$: we have $\hits_{Q(D)}(F)=1$ if and only if $D\in\set{R}\times
\Phi^{-1}(\set{v\colon R'(v)\in F})$ (and $\hits_{Q(D)}(F) = 0$ otherwise).

\begin{example}\label{ex:commonaggregates} The following are the most common 
	aggregate operators:
	\begin{itemize}
		\item (Count) $\CNT(\bag{a_1,\dots,a_n}) = n$ and $\CNTd(
			\bag{a_1,\dots,a_n})=\size{\set{a_1,\dots,a_n}}$.
		\item (Sum) $\SUM(\bag{a_1,\dots,a_n}) = a_1 + \dots + a_n$
			where $a_i$ are (for instance) real numbers.
		\item (Minimum\,/\,Maximum) $\MIN(\bag{a_1,\dots,a_n}) =
			\min\set{a_1,\dots,a_n}$ and $\MAX(\bag{a_1,\dots,a_n})
			= \max\set{a_1,\dots,a_n}$ for ordered domains.
		\item (Average) $\AVG(\bag{a_1,\dots,a_n}) = \frac{1}{n}
			\big(a_1+\dots+a_n\big)$ where the $a_i$ might again be 
			real numbers.
	\end{itemize}
\end{example}

Note that $\varpi_{\CNT}$ and $\varpi_{\CNTd}$ are trivially measurable within
our framework by the usage of the counting $\sigma$-algebra (and the 
measurability of deduplication for $\CNTd$).

\begin{lemma}\label{lem:aggaux}
	For all $m\in\N$, let $\phi_m\colon U^m\to V$ be a \emph{symmetric}
	function, i.\,e., $\phi_m(u) = \phi_m(u')$ for all $u\in U^m$
	and all permutations $u'$ of $u$. If $\phi_m$ is measurable for all 
	$m$, then $\Phi\colon\Bags{U}{<\omega}\to V$ defined via $\Phi(\bag{u_1,
	\dots,u_m}) \coloneqq \phi_m(u_1,\dots,u_m)$ is measurable w.\,r.\,t. the counting 
	$\sigma$-algebra on $\Bags{U}{<\omega}$.
\end{lemma}

\begin{proof}
	It suffices to show that the restriction $\Phi_m$ of $\Phi$ to 
	$\Bags{U}{m}$ is measurable for all $m\in\N$. If $\Event[V]$ is Borel
	in $V$, then $\phi_m^{-1}(\Event[V])$ is Borel in $U^m$ as $\phi_m$ is
	measurable.  Moreover, since $\phi_m$ is symmetric, $\phi_m^{-1}
	(\Event[V])$ is a symmetric set (i.\,e. if $\bar u\in\phi_m^{-1}
	(\Event[V])$, then every permutation of $u$ is in $\phi_m^{-1}
	(\Event[V])$ as well). But then $\Phi_m^{-1}(\Event[V])$ is measurable
	since there is a one-to-one correspondence between the measurable sets
	of $\Bags{U}{m}$ and the symmetric Borel sets of $U^m$ \cite[Theorem\,1]
	{Macchi1975}.
\end{proof}

As an example application of this lemma we note that all the mappings $\Phi$
that were introduced in \cref{ex:commonaggregates} are measurable---the related
mappings $\phi_m$ of \cref{lem:aggaux} are all continuous and thus measurable
in all of the cases.\par

A concept closely tied to aggregation is \emph{grouping}. Suppose we want to
group a relation $R$ by its attributes $A_1,\dots,A_k$ and perform the 
aggregation only over the values of attribute $A$, and separately for every
distinct $(A_1,\dots,A_k)$-entry in $R$. Without loss of generality, we assume
that the type of $R$ is $A_1\times\dots\times A_k\times A$. We define
a query $Q=\varpi_{A_1,\dots,A_k,\Phi(A)}(R)$ by
\begin{equation*}
	Q(D)=\bag{R'(\bar u,v) \colon
		R(\bar u)\in\pi_{A_1,\ldots,A_k}(R(D))~\text{and}~
	v=\Phi(\bag{u\colon R(\bar u,u)\in D})}\text.
\end{equation*}
                                                       
\begin{lemma}\label{lem:groupbyagg}
	Let $\type_{\Schema}(R)=A_1\times\dots\times A_k\times A$ and
	$U=\dom_{\Schema}(A)$. If $\Phi\colon\Bags{U}{<\omega}\to V$ is
	measurable (with $U$ and $V$ standard Borel), then 
	$\varpi_{A_1,\dots,A_k,\Phi(A)}(R)$ is a measurable query.
\end{lemma}

\begin{proof}
  Let $Q=\varpi_{A_1,\dots,A_k,\Phi(A)}(R)$ and
  $\bar A = (A_1,\dots,A_k)$. Observe that for every tuple
  $x_1, \dots,x_n,\epsilon$ with
  $x_i\in\prod_{j=1}^k\dom_{\Schema}(A_j)$ and $\epsilon > 0$, the
  following query is a composition of measurable queries and thus
  measurable itself:
	\begin{equation*}
		\tilde Q_{(x_1,\dots,x_n,\epsilon)} = 
		\textstyle\bigcup_{i=1}^n
		\project[\bar A]\big(\select[\bar A\in B_{\epsilon}(x_i)](R)
		\big)
		\product\varpi_\Phi\big(\project[A]
		\big(\select[\bar A\in B_{\epsilon}(x_i)](R)\big)
		\big)\text.
	\end{equation*}
	We have $\hits_{Q(D)}(F) = n$ if and only if there exist pairwise 
	distinct $f_1,\dots,f_n\in F$ such that $Q(D)$ has $1$ hit in each of
	the $f_i$ and nowhere else in $F$. Having $D$ fixed, every $f_i$
	determines the value of the $(A_1,\dots,A_k)$-part of an $R$-fact in
	$D$. Call this tuple $y_i$. We can fix a countable sequence of
	$(n+1)$-tuples $(x_1,\dots,x_n,\epsilon)$ such that (1) all $x_i$ are from a
	countable dense set in $\prod_{j=1}^k\dom_{\Schema}(A_j)$, (2) $d(x_i,y_i) <
	\epsilon$ for some fixed Polish metric, and, (3) $\epsilon\to 0$.
	Then $Q$ is the (pointwise) limit of the $\tilde Q_{(x_1,\dots,x_n,
	\epsilon)}$ and, as such, $Q$ is measurable.
\end{proof}

As noted before, the aggregates of \cref{ex:commonaggregates} easily satisfy
the precondition of \cref{lem:groupbyagg}.

\begin{corollary}
	The query $\varpi_{A_1,\dots,A_k,\Phi}(R)$ with $A_1,\dots,A_k\in
	\type_{\Schema}(R)$ is measurable for all aggregates $\Phi \in 
	\set{\CNT,\CNTd,\SUM,\MIN,\MAX,\AVG}$.
\end{corollary}

\section{Datalog Queries}\label{sec:datalog-queries}
In this section, we want to show that our measurability results
extend to datalog queries and in fact all types of queries with
operators based on countable
iterative (or inductive, inflationary, fixed-point) processes. We will
not introduce datalog or any of the related query languages. The
details in the definitions do not matter when it comes to
measurability of the queries. Here, we only consider set PDBs and queries with a set
(rather than bag) semantics. The key observation is the following lemma.

\begin{lemma}\label{lem:countable-union}
  Let $Q_i$, for $i\in\N_+$, be a countable family of measurable
  queries of the same schema such that $Q=\bigcup_{i\ge 1}Q_i$,
  defined by $Q(D)\coloneqq\bigcup_{i\ge 0}Q_i(D)$ for every instance $D$, is
  a well-defined query (that is, $Q(D)$ is finite for every $D$). Then
  $Q$ is measurable.
\end{lemma}

\begin{proof}
  For every $n\in\N_+$, let $Q^{(n)}\coloneqq \bigcup_{i=1}^nQ_i$. As a finite
  union of measurable queries, $Q^{(n)}$ is measurable. Since $Q=\lim_{n\to
  \infty} Q^{(n)}$, the measurability of $Q$ follows.
\end{proof}

As every datalog query can be written as a countable union of
conjunctive queries, we obtain the following corollary.

\begin{corollary}
  Every datalog query is measurable.
\end{corollary}

The same is true for queries in languages like inflationary datalog or least
fixed-point logic. For partial datalog / fixed-point logic, we cannot
directly use Lemma~\ref{lem:countable-union}, but a slightly more
complicated argument still based on countable limits works there
as well.

\section{Beyond Possible Worlds Semantics}
\label{sec:beyond-possible-worlds-semantics}
In the literature on probabilistic databases, and motivated by real world
application scenarios, also other kinds of queries have been investigated that
have no intuitive description in the possible worlds semantics framework. A
range of such queries is surveyed in \cite{Aggarwal+2009,Wang+2013}. The 
reason for the poor integration into possible worlds semantics is because such
queries lack a sensible interpretation on single instances that could be lifted
to PDB events. Instead, they directly refer to the probability
space of all instances.

Notable examples of such queries (cf. \cite{Koch+2008,Aggarwal+2009,Wang+2013}) 
are:
\begin{itemize}
	\item \emph{probabilistic threshold queries} that intuitively return
		a deterministic table containing only those facts which have a 
		marginal probability over some specified threshold;
	\item \emph{probabilistic top-$k$-queries} that intuitively return a
		deterministic table containing the $k$ most probable facts; 
	\item \emph{probabilistic skyline queries}~\cite{Pei+2007} that consider
		how different instances compare to each other with respect to
		some notion of \emph{dominance}; and
	\item \emph{conditioning}~\cite{Koch+2008} the
		probabilistic database to some event.
\end{itemize}
Note that the way we informally explained the first two queries above is only 
sensible if the space of facts is discrete. In a continuous 
setting, we interpret these queries with respect to a suitable countable 
partition of the fact space into measurable sets.

Let $\PDBs_{\Schema}$ denote the class of probabilistic databases of schema
$\Schema$. Note that all PDBs in $\PDBs_{\Schema}$ have the same instance
measurable space $(\Instances,\Measurable)$. Queries and, more
generally, views of
input schema $\Schema$ and output schema $\Schema'$ are now mappings
$V\colon\PDBs_{\Schema}\to\PDBs_{\Schema'}$.\par

We classify views in the following way:
\begin{definition}\label{def:types}
	Let $V\colon\PDBs_{\Schema}\to\PDBs_{\Schema'}$ with $V\colon
	\pdb=(\Instances,\Measurable,P)\mapsto(\Instances',\Measurable',P')=
	\pdb'$. 
	\begin{enumerate}
		\item Every view $V$ is of \emph{type I}.
		\item The view $V$ is of \emph{type II} (or, \emph{pointwise
			local}) if for every $\pdb\in\PDBs_{\Schema}$ there
			exists a measurable mapping $q_{\pdb}\colon 
			\Instances\to\Instances$ such that $P'(\Event') = 
			P(q_{\pdb}^{-1} (\Event'))$ for every $\Event'\in
			\Measurable$.
		\item\label{itm:typeIII}
			The view $V$ is of \emph{type III} (or, \emph{uniformly
			local}) if there exists a measurable mapping
			$q\colon\Instances \to\Instances$ such that
			$P'(\Event') = P(q^{-1}(\Event'))$ for every
			$\Event'\in\Measurable'$.
	\end{enumerate}
\end{definition}

\newcommand*{\viewtype}[1]{\cal{V}^{\text{#1}}}

Letting $\viewtype{I}$, $\viewtype{II}$ and $\viewtype{III}$ denote the classes
of type I, type II and type III views (from $\PDBs_{\Schema}$ to 
$\PDBs_{\Schema'}$).
Then $\viewtype{III}$ captures the possible worlds semantics of views.
Obviously, $\viewtype{III} \subseteq \viewtype{II} \subseteq
\viewtype{I}$. The following examples show that these inclusions are strict.

\begin{example}
	Consider the query $Q = Q_\alpha(D) = \set{ f \in \facts_{\Schema}(R) 
	\colon P(\CEvent(f,>0)) \geq \alpha} = q_\pdb$ for some
      $\alpha > 0$. Note that the set of facts of marginal probability
      at least $\alpha$ is finite in every PDB \cite{Grohe+2019}, hence the
      query is well-defined.
	This query is of type II. However, considering the simple PDBs $\pdb_1$
	and $\pdb_2$ and two distinct facts $f$ and $f'$ such that
	\begin{itemize}
		\item the only possible world of positive probability in
			$\pdb_1$ is $\bag{ f }$ with $P_{\pdb_1}(\bag{ f })=1$;
		\item similarly, $\pdb_2$ has the worlds $\bag{ f }$ and 
			$\bag{ f' }$ with $P_{\pdb_2}(\bag{ f }) =
			P_{\pdb_2}(\bag{ f' }) = \frac{1}{2}$.
	\end{itemize}
	Suppose $q$ exists like in the \cref{def:types}, part 
	\labelcref{itm:typeIII} and consider the event $\Event'$ that $f'$
	occurs (this is a set of instances in the target measurable space of
	$Q_\alpha$). Then $P_{\pdb_1}(q^{-1}(\Event')) = 0$ entails 
	$\bag{f} \notin q^{-1}(\Event')$. On the other hand $P_{\pdb_2}
	(q^{-1}(\Event')) = 1$ and thus $\bag{f},\bag{f'}\in q^{-1}(\Event')$,
	a contradiction. Thus, $Q$ is type II, but not type III.\par
\end{example}
\begin{example}
 	Fix some PDB $\pdb$ with three possible worlds $D_1$, $D_2$ and
	$D_3$ with probabilities $p_1 = \frac{1}{6}$, $p_2 = \frac{1}{3}$ and
	$p_3 = \frac{1}{2}$. Now consider the query $Q$ that conditions
	$\pdb$ on the event $\set{ D_1, D_2 }$ and pick the database instance
	$D = D_1$. Then $P(D\cap\set{D_1,D_2}) = P(\set{D_1}) = \frac{1}{6}$ and 
	$P(\set{D_1,D_2}) = \frac{1}{6}+\frac{1}{2} = \frac{4}{6}$. Thus,
	$P(Q^{-1}(D)) = \frac{1}{6} / \frac{4}{6} = \frac{1}{4}$, but there is
	no event $\Event$ in $\pdb$ with the property that
	$P( \Event ) = 1/4$. Thus, $Q$ is type I, but not type II.
\end{example}

\section{Conclusions}\label{sec:conclusion}
In this work, we described how to construct suitable probability spaces for
infinite probabilistic databases, completing the picture of \cite{Grohe+2019}.
The viability of this model as a general framework for finite \emph{and 
infinite} databases is supported by its compositionality with respect to
typical database queries. Our main technical results establish that standard query languages
have a well-defined open-world semantics.\par

It might be interesting to explore, whether more in-depth results on point
processes have a natural interpretation when it comes to probabilistic 
databases. We believe for example that there is a strong connection between
the infinite independence assumptions that were introduced in \cite{Grohe+2019}
and the class of Poisson point processes (cf. \cite[p.~52]{Last+2017}).\par
In the last section of the paper, we briefly discussed queries for PDBs that go
beyond the possible worlds semantics. Such queries are very relevant for PDBs
and deserve a systematic treatment in their own right in an infinite 
setting.

\subparagraph*{Acknowledgments}\label{sec:acknowledgments}
We are grateful to Sam Staton for insightful discussions related to
this work, and for pointing us to point processes. We also thank Peter J.
Haas for discussions on the open-world assumptions and the math behind the MCDB system.
\phantomsection\label{sec:bibliography}

\clearpage\appendix

\section{Notation}\label{app:notation}
\begin{table}[H]
	\centering
	\caption{Basic notation used throughout the paper.}
	\begin{tabular}{p{.22\textwidth} p{.7\textwidth}}
		\toprule
		\multicolumn{2}{l}{%
			\bfseries{}\sffamily{}%
			General Notation%
		}\\
		\midrule
		$\N$, $\Q$, $\R$
		& the sets of nonnegative integers, rational numbers and
		real numbers, respectively\\
		$\N_+$, $\Q_+$, $\R_+$
		& the restrictions of $\N$, $\Q$ and $\R$ to positive numbers\\
		$\Sets{M}{k}$, $\Sets{M}{<\omega}$
		& the sets of $k$-elementary respectively finite \emph{subsets}
		of a set $M$\\
		$\bag{\dots}$
		& the explicit denotation of a bag / multiset\\
		$\Bags{M}{k}$, $\Bags{M}{<\omega}$
		& the sets of $k$-elementary respectively finite \emph{bags}
		over a set $M$\\
		$\hits_N(\cdot)$
		& multiplicity function of a bag $N$\\
		\midrule
		\multicolumn{2}{l}{%
			\bfseries{}\sffamily{}%
			Topology, Measure Theory and Point Processes%
		}\\
		\midrule
		$\Space[X], \Space[Y], \Space[Z]$
		& underlying set of a probability / measurable / topological space\\
		$X, Y, Z$
		& an element of $\Space[X]$, $\Space[Y]$ or $\Space[Z]$,
		respectively\\
		$x, y, z$
		& elements of $X$, $Y$ or $Z$, respectively (provided that
		$X$, $Y$, $Z$ are (subsets of) the powerset of some other 
		space)\\
		$\Measurable[X], \Measurable[Y], \Measurable[Z]$
		& a $\sigma$-algebra on  $\Space[X]$, $\Space[Y]$ or 
		$\Space[Z]$, respectively\\
		$\Measurable[G]$
		& used for a generating family of a $\sigma$-algebra\\
		$\Event[X], \Event[Y], \Event[Z]$
		& sets (but usually \emph{measurable} sets) in $(\Space[X],
		\Measurable[X])$, $(\Space[Y],\Measurable[Y])$ or $(\Space[Z],
		\Measurable[Z])$\\
		$P$
		& (probability) measures on some measurable space $(\Space[X],
		\Measurable[X])$\\
		$\Xi$
		& a probability space $\Xi = (\Space[X],\Measurable[X],P)$\\
		$\Topology$
		& a topology on some space $\Space[X]$\\
		$B_\epsilon(x)$
		& \enquote{ball} of radius $<\epsilon$ around a point $x$ of a metric
		space\\
		$\Borel_{\Space[X]}$
		& the Borel $\sigma$-algebra induced by some topological
		space $(\Space[X],\Topology)$\\
		\midrule
		\multicolumn{2}{l}{%
			\bfseries{}\sffamily{}%
			Databases and Probabilistic Databases%
		}\\
		\midrule
		$\Schema$
		& a database schema $\Schema = (\Att, \Rel)$\\
		$\Att$
		& a family of attribute names $A$\\
		$A,B$
		& an attribute (name)\\
		$\Rel$
		& a family of relation names $R$\\
		$R,S$
		& relations / relation names\\
		$\dom_{\Schema}(\cdot),\type_{\Schema}(\cdot),\ar_{\Schema}(\cdot)$
		& the domain, type and arity mappings of a schema $\Schema$\\
		$\facts_{\Schema}(R),\facts_{\Schema}(\Rel)$
		& the set of $R$-facts, resp. \emph{all} facts, in schema $\Schema$\\
		$D$
		& a database instance\\
		$f$
		& a single fact\\
		$F$
		& a set of facts\\
		$\Instances$, $\Instances_{\Schema}$ 
		& the space of database instances w.\,r.\,t. some schema
		$\Schema$\\
		$\Measurable$,  $\Measurable_{\Schema}$
		& the $\sigma$-algebra belonging to $\Instances$ resp. $\Instances_{\Schema}$\\
		$\CEvent(F,n)$
		& the counting event belonging to $F$ and $n$\\
		$\Delta$
		& a probabilistic database $\Delta=(\Instances,\Measurable,P)$\\
		$\PDBs_{\Schema}$
		& the class of PDBs of schema $\Schema$\\
		$Q$
		& a query\\
		$V$
		& a view\\
		$\rename,\addunion,\intersection,\difference,\maxunion,\dedupe,\select,\project,\product$
		& bag relational algebra operators, see \cref{tab:balg} (p.~\pageref{tab:balg})\\
		$ \Phi$
		& an aggregator\\
		$ \varpi$
		& an aggregate query (possibly with grouping)\\
		\bottomrule
	\end{tabular}
\end{table}
\section{Notions from General Topology}\label{app:topology}
In this section, we introduce the relevant topological notions that are needed
in our work. The reader may find further reference in textbooks on general
topology such as \cite{Bourbaki1995}. Polish spaces are in particular discussed
within \cite{Bourbaki1989}. \par\smallskip

A \emph{topological space} is a pair $(X,\Topology)$ where $X$ is a set and 
$\Topology$ is a family of subsets of $X$ such that
\begin{itemize}
	\item both $\emptyset$ and $X$ belong to $\Topology$;
	\item $\Topology$ is closed under \emph{arbitrary} unions; and
	\item $\Topology$ is closed under \emph{finite} intersections.
\end{itemize}
Such a family $\Topology$ is called \emph{topology on $X$} and its individual sets
are called \emph{open sets}. Complements (relative to $X$) of open sets are
called \emph{closed}. Occasionally, we might refer to a topological space $X$,
if the topology on $X$ is clear from context.\par

If $\Topology$ and $\Topology'$ are topologies on $X$, then $\Topology$ is called
\emph{coarser} than $\Topology'$ if $\Topology\subseteq \Topology'$. Vice versa, 
$\Topology'$ is called \emph{finer}. A mapping between two topological spaces is 
called \emph{continuous}, if the preimage of every open set is open. It is
called \emph{open}, if it maps open sets to open sets.\par

Whenever $((X_i,\Topology_i))_{i\in I}$ is a family of topological spaces, then
the \emph{product topological space} (or \emph{product topology}) of 
$((X_i,\Topology_i))_{i\in I}$ is the (unique) coarsest topology $(X,\Topology)$ with
$X=\prod_{i\in I} X_i$ such that all the canonical projection maps $\proj_{i\in
I}\colon X\to X_i$ are continuous.\par\smallskip

A \emph{metric space} is a pair $(X,d)$ where $X$ is a set and $d\colon X\to\R$
such that for all $x,y,z\in X$
\begin{itemize}
	\item $d(x,y)\geq 0$; and $d(x,y)= 0$ if and only if $x=y$;
	\item $d(x,y) = d(y,x)$; and
	\item $d(x,y)\leq d(x,z)+d(z,y)$.
\end{itemize}
In a metric space $(X,d)$, for $x\in X$ and $r\in\R_+$, we denote by $B_r(x)$ 
the \emph{open ball of radius $r$ around $x$}, that is, the set $B_r(x) 
\coloneqq \set{ y \in X \colon d(x,y) < r }$.\par

A set $Y\subseteq X$ is called \emph{open}, if for every $y\in Y$, there is 
some $r>0$ such that $B_r(y)\subseteq Y$. Complements of open sets are called
\emph{closed}. The open sets of the metric space $(X,d)$ form a topology on
$X$, which we refer to as the topology on $X$ that is \emph{induced} (or 
\emph{generated}) by $d$. A topological space is called \emph{metrizable} if
it can be equipped with some metric that generates its topology.\par

A \emph{Cauchy sequence} in a metric space $(X,d)$ is a sequence $(x_k)_{k\geq 
0}$ of elements of $X$ with the property that $d(x_k,x_{k+1})\to 0$ as $k\to
\infty$. The metric space $(X,d)$ is called \emph{complete} if $\lim_{k\to
\infty} x_k\in X$ for every Cauchy sequence $(x_k)_{k\geq 0}$ in $X$. A
topological space is called \emph{completely metrizable} if it can be equipped
with some \emph{complete} metric that generates its topology.\par

A set $Y$ in a topological space $(X,\Topology)$ (or metric space $(X,d)$) is 
called \emph{dense} if for every $x\in X$ either $x\in Y$ or there are is a
sequence $(y_k)_{k\geq 0}$ in $Y$ with $\lim_{k\to\infty} y_k = x$. A
topological space (or metric space) is called \emph{separable} if it contains a
countable dense set. A \emph{Polish space} is a separable, completely
metrizable topological space. In particular, separable, complete 
\emph{metric} spaces are Polish. We call metrics of such spaces \emph{Polish}
metrics and refer to the topology of a Polish space as its \emph{Polish 
topology}.

\section{Measurability of Cross Products}\label{app:crossprod}

\begin{lemma}\label{lem:crossprod-measurable-app}
$Q=R_1\product{}R_2$ is $(\Measurable,\Measurable')$-measurable for all
$R_1,R_2\in\Rel$ with disjoint types.
\end{lemma}
\par\bigskip

We split the proof in several parts. First we show the following claim, stating
that preimages of counting events that are products of Borel sets are
measurable.\par\bigskip

\begin{claim}\label{clm:rectangles-multiply}
	For all $n\in\N$, it is $Q^{-1}(\CEvent'(F,n))\in\Measurable$
	whenever $F = \set{ R(t_1,t_2)\colon t_1\in T_1,t_2\in T_2}$ for
	Borel sets $T_1$ and $T_2$ of $\dom_{\Schema}(R_1)$ resp.
	$\dom_{\Schema}(R_2)$.
\end{claim}

\begin{subproof}
	Let $T_i\in\Borel(\dom_{\Schema}(R_i))$ and let
	$F_i=\set{R_i (t)\colon t\in T_i}$ for $i\in\set{1,2}$. By the
	semantics of $Q$ and since $F=\set{R(t_1,t_2)\colon{}R_1
	(t_1)\in F_1~\text{and}~R_2 (t_2)\in F_2}$ it holds that
	$\hits_{Q(D)}(F) = \hits_{D}(F_1)\cdot\hits_{D}(F_2)$
	for all instances $D\in\Instances$. Thus,
	\begin{equation*}
		D\in Q^{-1}(\CEvent'(F,n)) \iff
		\begin{aligned}[t]
			&\text{there exist}~n_1,n_2\in\N~
				\text{with}~n_1\cdot n_2=n~\text{such that}\\
			&\quad	D\in\CEvent(F_1,n_1)~\text{and}~D\in
				\CEvent(F_2,n_2)\text.
		\end{aligned}
	\end{equation*}
	In particular, $Q^{-1}(\CEvent'(F,n))\in\Measurable$.
\end{subproof}
\par\bigskip

Now that we have established that the preimage of every product of Borel sets is
measurable, we turn our attention to the general setting. We fix countable dense
sets $X$, $X_1$ and $X_2$ in the Polish spaces $\dom_{\Schema'}(R)$,
$\dom_{\Schema}(R_1)$ and $\dom_{\Schema}(R_2)$. Also, we let
$d$, $d_1$ and $d_2$ denote fixed, Polish metrics on the aforementioned spaces.
For $D\in\Instances$, define
\begin{align*}
	d_i^*(D) \coloneqq \min\set{d_i(t,t')\colon\hits_D(R_i(t)),
	\hits_D(R_i(t'))>0~\text{and}~ t\neq t'}\text.
\end{align*}
If the set on the right is empty, we let $d_i^*(D)=\infty$. For the purpose of
this proof, we refer to instances $D\in\Instances$ as \emph{$k$-coarse} if both
$d_1^*(D)$ and $d_2^*(D)$ are at least $k^{-1}$. We let
\begin{equation}\label{eq:DFnk-event}
	\Event(F,n,k) \coloneqq Q^{-1}(\CEvent'(F,n))\cap
		\set{D\in\Instances\colon D~\text{$k$-coarse}}
\end{equation}
denote the restriction of the preimage of $\CEvent'(F,n)$ to $k$-coarse 
instances. Note 
\begin{equation*}
	Q^{-1}(\CEvent'(F,n))=\bigcup_{k\in\N}\Event(F,n,k)\text.
\end{equation*} Thus, we are done, once we have proven the measurability of 
$\Event(F,n,k)$ for all $F\in\Measurable[F]'$ and $n,k\in\N$, $k>0$.\par\bigskip

Conceptually, for certain \emph{simple} sets $F$ of facts, we will prove the 
measurability of $\Event(F,n,k)$ directly from the measurability of 
$Q^{-1}(\CEvent'(F,n))$ and $\set{D\in\Instances\colon D~
\text{$k$-coarse}}$ and proceed to show that we can obtain the measurability of 
$\Event(F,n,k)$ for arbitrary $F$ from these simple cases.\par\bigskip

We start out by proving that the set of $k$-coarse instances is $\Measurable$-%
measurable.

\begin{claim}\label{clm:coarse-measurable}
	The set $\set{D\in\Instances\colon D~\text{$k$-coarse}}$ of $k$-coarse
	instances in $\Instances$ is $\Measurable$-measurable.
\end{claim}

\begin{subproof}
	We claim (for $i\in\set{1,2}$):
	\begin{equation}\label{eq:k-coarse-parts}
		d_i^*(D) < \frac{1}{k}
		\quad\iff\quad
		\begin{aligned}[t]
			& \text{there exist}~k,k'\in\N_+~\text{and}~\epsilon_L,
				\epsilon_U\in\Q_+\text,~\epsilon_U<\epsilon_L<
				\textstyle\frac{1}{k}~\text{such that}\\
			&\quad \text{for all}~r\in\Q_+~\text{where}~r<\min
				\set{\textstyle\frac{1}{3}(\frac{1}{k}-
				\epsilon_L),\frac{\epsilon_U}{4}}\\
			&\quad\quad \text{there are}~x,x'\in X_i~
				\text{with}~0<d(x,x')<\textstyle\frac{1}{k}-
				\epsilon_U+2r~\text{such that}\\
		       &\quad\quad\quad \hits_D(B_r(x))=k~\text{and}~
		       		\hits_D(B_r(x'))=k' 
		\end{aligned} 
	\end{equation} 
	where $B_r(x)=\set{R_i(t)\colon d_i(t,x)<r}$. Note that the
	property on the right hand side of \labelcref{eq:k-coarse-parts} is
	expressible as a \enquote{countable Boolean combination} of counting
	events (the set of instances satisfying it is Borel). Since
	\begin{equation*}
		\set{D\in\Instances\colon D~\text{is $k$-coarse}} =
		\Instances \setminus \set{D\in\Instances\colon d_1^*(D)<\textstyle
			\frac{1}{k}~\text{or}~d_2^*(D)<\frac{1}{k}}\text,
	\end{equation*}
	we are done once we have demonstrated \labelcref{eq:k-coarse-parts}.
	We do so by showing both directions. Note that the rationale for $i=1$
	and $i=2$ is identical.

\begin{description}
	\item[$(\Rightarrow)$] 
		Let $d_i^*(D)<\frac{1}{k}$ and let $f$ and $f'$ be
		$R_1$-facts appearing in $D$ such that $d_i(f,f') <
		\frac{1}{k}$ is minimal (since $0 < d_i^*(D) < \frac{1}{k}$, $f$
		and $f'$ exist). Let $0 < \epsilon_U < \epsilon_L < \frac{1}{k}$
		be rational numbers with $d_i(f,f') \in (\frac{1}{k} -
		\epsilon_L, \frac{1}{k} - \epsilon_U)$. Let $k=\hits_D(f)$ and
		$k'=\hits_D(f')$. Since $X_i$ is dense, for every positive $r$
		(so in particular for all rational $r < \min\set{\frac{1}{3}
		(\frac{1}{k} - \epsilon_L, \frac{\epsilon_U}{4}}$) there exist
		$x,x' \in X_i$ such that $f\in B_r(x)$ and $f'\in B_r (x')$.
		By the choice of $r$, $B_r(x)$ and $B_r(x')$ are disjoint. In
		particular $d_i(x,x')>0$. Also, since $f$ and $f'$ are a pair
		of $R_1$-facts of $D$ of minimal distance, there are no
		other $R_1$-facts (other than $f$ resp. $f'$) that are
		contained in $B_r(x)$ resp. $B_r(x')$. Thus,
		\begin{equation*}
			k  = \hits_D(f)  = \hits_D(B_r(x))
			\qquad\text{and}\qquad
			k' = \hits_D(f') = \hits_D(B_r(x'))\text.
		\end{equation*}
		Moreover, the distance of $x$ and $x'$ is necessarily smaller 
		than $\frac{1}{k}-\epsilon_U+2r$:
		\begin{equation*}
			d_i(x,x')
			\leq d_i(x,f)+d_i(f,f')+d_i(f',x')
			< r + \textstyle\frac{1}{k}-\epsilon_U + r
			= \textstyle\frac{1}{k}-\epsilon_U+2r\text.
		\end{equation*}
		On the other hand (by a similar application of the triangle
		inequality), $d_i(x,x') > r > 0$. Overall, the right hand side 
		of \labelcref{eq:k-coarse-parts} holds.

	\item[$(\Leftarrow)$]
		Now suppose the right hand side of
		\labelcref{eq:k-coarse-parts} holds and let $k,k',\epsilon_L$
		and $\epsilon_U$ such that the rest of the statement is
		satisfied.\par

		Now for all positive rational $r < \min\set{\frac{1}{3}
		(\frac{1}{k} - \epsilon_L), \frac{\epsilon_U}{4}}$ there exist
		$x$ and $x'$, $x\neq x'$ from $X_i$ with distance smaller than
		$\frac{1}{k} - \epsilon_U + 2r$ such that $\hits_D (B_r(x)) =
		k$ and $\hits_D(B_r(x')) = k'$. This holds in particular, if
		additionally $r < d_i^*(D)/3$. For such $r$, the balls $B_r(x)$
		and $B_r(x')$ contain at most one $R_1$-fact from $D$
		each, say $f$ in $B_r (x)$ and $f'$ in $B_r(x')$. Then $k =
		\hits_D(B_r(x)) = \hits_D(f)$ and $k' = \hits_D(B_r(x')) =
		\hits_D(f')$ and both are $>0$.\par

		It is 
		\begin{equation*}
			d(f,f')
			\leq d(x,f)+d(x,x') + d(f',x')
			< r + (\textstyle \frac{1}{k} - \epsilon_U + 2r) + r
			< \textstyle \frac{1}{k}
		\end{equation*}
		where the last inequality is due to $r < \epsilon_U/4$.	
		Together, $f$ and $f'$ witness $d_i^*(D)<\frac{1}{k}$.\qedhere
      \end{description}
\end{subproof}
\par\bigskip

For $k\in\N$ we let 
\begin{equation*}
	\Event[F]_k \coloneqq
	\set*{F\in\Measurable[F]_{\Schema'}(R)\colon
		\Event(F\cap S_r(t),n,k)\in\Measurable~\text{f.\,a.}~
		t\in\dom_{\Schema'}(R)\text,~r\in\Q_+\text,
		~r<\textstyle\frac{1}{3k}~\text{and}~n\in\N} 
\end{equation*}
where $S_r(t) \coloneqq \set{R} \times B_r^{(1)}(t_1) \times B_r^{(2)}
(t_2)$ is the rectangle around $t = (t_1,t_2) \in \dom_{\Schema}(R_1)
\times \dom_{\Schema} (R_2)$ whose \enquote{sides} are given by the
balls $B_r^{(i)}(t_i)$ of $d_i$-radius $r$ around $t_i$ in $\dom_{\Schema}
(R_i)$.\par\bigskip

\begin{claim}
	$\Event[F]_k=\Measurable[F]_{\Schema'}(R)$ for all $k\in\N$.
\end{claim}

\begin{subproof}
	Fix some $k\in\N$. We demonstrate $\Event[F]_k=\Measurable[F]_{\Schema'}
	(R)$ in the following two steps (using the \emph{good sets
	principle} \cite{Ash1972}):
	\begin{enumerate}
		\item\label{itm:rectangles-are-good}
			Let $F$ be a rectangle like in
			\cref{clm:rectangles-multiply} (that is,
			$F=\set{R(t_1,t_2)\colon t_1\in T_1,t_2\in
			T_2}$ for some measurable sets $T_1,T_2$ of
			$R_1$- resp $R_2$-tuples) belong to
			$\Event[F]_k$.
		\item\label{itm:good-sets-sigma-algebra} 
			$\Event[F]_k$ is a $\sigma$-algebra on
			$\facts_{\Schema'}(R)$, more precisely, we show
			that $\Event[F]_k$ contains $\facts_{\Schema'}
			(R)$ and is closed under complement, finite 
			intersection and countable disjoint union.%
			\footnote{If a family of subsets is closed under
			complement, finite intersection and \emph{disjoint}
			countable union, then it is closed under general
			countable unions as well. We proceed this way as it
			feels more natural to argue about (finite)
			intersections and disjoint unions in the present
			setting.}
	\end{enumerate}
	Since the sets from \ref{itm:rectangles-are-good} generate 
	$\Measurable[F]_{\Schema'}(R)$, \ref{itm:rectangles-are-good} and
	\ref{itm:good-sets-sigma-algebra} together will imply that $\Event[F]_k=
	\Measurable[F]_{\Schema'}(R)$.\par

	\begin{enumerate}
		\item\label{itm:Fk-rectangles}
			Let $F=\set{R(t_1,t_2)\colon{}R_1(t_1)\in
			F_1~\text{and}~ R_2(t_2)\in F_2}$ for some
			$F_i\in\Measurable[F]_{\Schema}(R_i)$ for
			$i\in\set{1,2}$.  Fix an arbitrary
			$t\in\dom_{\Schema'}(R)$ with $t=(t_1, t_2)$
			where $t_1\in\dom_{\Schema}(R_1)$ and $t_2\in
			\dom_{\Schema}(R_2)$. Let $r\in\Q_+$ with $r <
			\frac{1}{3k}$. Then 
			\begin{equation*}
				F\cap S_r(t) = 
				(F_1\times F_2)\cap
					(B_r^{(1)}(t_1)\times B_r^{(2)}(t_2)) = 
				(F_1\cap B_r^{(1)}(t_1))\times 
					(F_2\times B_r^{(2)}(t_2))
			\end{equation*} 
			is a measurable rectangle itself. Thus,
			$Q^{-1}(\CEvent'(F\cap S_r(t)))$ is
			measurable by \cref{clm:rectangles-multiply}. Since 
			$\set{D\in\Instances\colon D~\text{is $k$-coarse}}$ is 
			measurable by \cref{clm:coarse-measurable}, also the 
			intersection 
			\begin{equation*}
				Q^{-1}(\CEvent'(F\cap S_r(t),n))
					\cap\set{D\colon D~\text{$k$-coarse}} =
				\Event(F\cap S_r(t),n,k)
			\end{equation*} 
			(cf.~\cref{eq:DFnk-event}) is measurable. Thus, $F\in
			\Event[F]_k$ follows.

		\item We show that $\Event[F]_k$ is a $\sigma$-algebra on
			$\facts_{\Schema'}(R)$. In the following let
			$X$ be a countable dense set in $\dom_{\Schema'}
			(R)$.
			\begin{itemize}
				\item $\facts_{\Schema'}(R)\in
					\Event[F]_k$ follows from item
					\labelcref{itm:Fk-rectangles} above, 
					since $\facts_{\Schema'}(R)=
					\facts_{\Schema}(R_1)\times
					\facts_{\Schema}(R_2)$ is a 
					measurable rectangle.
				\item Let $F\in\Event[F]_k$ and consider its complement 
					$\complement{F}$. Since
					$F\in\Event[F]_k$, for all $x\in X$,
					all $r\in\Q_+$ with $r < \frac{1}{3k}$
					and all $n \in \N$ it holds that
					$\Event(F\cap S_r(x),n,k)\in\Measurable$.
					We fix such $x$, $r$ and $n$
					arbitrarily. Then 
					\begin{align*}
						&D\in\Event(\complement{F}\cap
						S_r(x),n,k)\\
						\iff&
						D~\text{is $k$-coarse and}~
						\hits_{Q(D)}
						(\complement{F}\cap S_r(x))=n\\
						\iff&
						D~\text{is $k$-coarse and}~
						\hits_{Q(D)}(S_r(x))
						-\hits_{Q(D)}
						(F\cap S_r(x)) = n\\
						\iff&
						D\in\Event(S_r(x),n_1,k) \cap 
						\Event(F\cap S_r(x),n_2,k)\\
						&\qquad\text{for some}~n_1,n_2
						\in\N~\text{with}~n_1-n_2=n
						\text.
					\end{align*}
					Thus, $\complement{F}\in\Event[F]_k$.
					Note that $Q(D)$ contains at
					most one fact in $S_r(x)$ since $D$ is
					$k$-coarse.
				\item Now let $F_1,F_2\in\Event[F]_k$.
					Similarly, $\Event(F_1\cap S_r(x),
					n,k)$ and $\Event(F_2\cap S_r(x),n,k)$
					are $\Measurable$-measurable for all $x\in
					X$, rational $0<r<\frac{1}{3k}$ and
					$n\in\N$.  Again, we fix such $x,r$ and
					$n$. Then 
					\begin{align*}
						&D\in\Event((F_1\cap F_2)\cap
						S_r(x),n,k)\\
						\iff&
						D~\text{is $k$-coarse and}~
						\hits_{Q(D)}
						((F_1\cap F_2)\cap S_r(x))=n\\ 
						\iff&
						D~\text{is $k$-coarse and}\\
						&\qquad\min\set{
						\hits_{Q(D)}(F_1\cap
						S_r(x)),\hits_{Q(D)}
						(F_2\cap S_r(x))}=n\\
						\iff&
						D\in\Event(F_1\cap S_r(x),n_1,k)
						\cap\Event(F_2\cap S_r(x),n_2,k)
						\\
						&\qquad\text{for some}~n_1,n_2
						\in\N~\text{with}~\min\set{n_1,
						n_2}=n\text.
					\end{align*}
					Thus\footnote{In the equivalences, note
					that the straight-forward assertion 
					\enquote{$n$ hits in $S_r(x)$ and $0$
					hits in $F\cap S_r(x)$} fails for
					$n=0$ in the case where $Q
					(D)$ has hits in $S_r(x)$ but all of 
					them are in $F\cap S_r(x)$.}, $F_1\cap 
					F_2\in\Event[F]_k$. Again, the second 
					equivalence above holds because $D$ 
					contains at most one fact from $S_r(x)$.
		
				\item Finally, let $F_i\in\Event[F]_k$ for
					$i\geq 0$ such that the $F_i$ and $F_j$
					are disjoint for $i\neq j$.  For every
					$i\geq 0$ for all $x\in X$, all
					$r\in\Q_+$ with $r<\frac{1}{3k}$ and
					all $n\in\N$, $\Event(F_i\cap
					S_r(x),n,k)\in\Measurable$. Now, once
					again, fix $x$, $r$ and $n$ as
					specified. Then 
					\begin{align*}
						&D\in\Event((\textstyle
						\bigcup_i F_i)\cap S_r(x),n,k)\\
						\iff&
						D~\text{is $k$-coarse and}~
						\hits_{Q(D)}
						((\textstyle\bigcup_i F_i)\cap 
						S_r(x))=n\\
						\iff&
						D~\text{is $k$-coarse and}~
						\hits_{Q(D)}
						(\textstyle\bigcup_i (F_i\cap 
						S_r(x)))=n\\
						\iff&
						\begin{aligned}[t]
							\text{there is}~i\in\N~
							\text{with}~
							\hits_{Q(D)}
							(F_i\cap S_r(x))={}&n\\
							\text{s.\,t. for all}~
							j\neq i\colon
							\hits_{Q(D)}
							(F_j\cap S_r(x))={}&0
							\text.
						\end{aligned}
					\end{align*}
					We obtain\footnote{Like before, the
					straight-forward \enquote{$n$ hits in
					$F_1\cap S_r(x)$ and $n$ hits in $F_2
					\cap S_r(x)$} would fail for $n=0$ in
					the case when there are hits in $S_r(x)$
					but all of them are in $(F_1\setminus
					F_2)\cap S_r(x)$ (or the other way
					around).} $\bigcup_i F_i\in\Event[F]_k$.
					Again, we used that $D$ and $S_r(x)$
					have at most one distinct fact in
					common.\qedhere 
			\end{itemize} 
	\end{enumerate}
\end{subproof}

\begin{claim}\label{clm:DFnk-measurable}
	For all $F\in\Measurable[F]'$, $n,k\in\N$, the set $\Event(F,n,k)$ is
	$\Measurable$-measurable.
\end{claim}

\begin{subproof}
	Let $\Event^\geq(F,n,k)\coloneqq\bigcup_{n'\geq n}\Event(F,n',k)$. Then
	\begin{equation*}
		\Event(F,n,k)=\Event^\geq(F,n,k)\cap\complement{(\Event^\geq(F,n+1,k))}
		\text.
	\end{equation*}
	Thus, it suffices to show the measurability of $\Event^\geq(F,n,k)$ for
	all $n\in\N$ to show \cref{clm:DFnk-measurable}. We claim 
	$D\in\Event^\geq(F,n,k)$ if and only if
	\begin{align*}
		D~\text{is $k$-coarse and}~
			&\text{there exist}~m,k_1,\dots,k_m~\text{with}~
				\textstyle\sum_{i=1}^mk_i\geq n~ \text{such 
				that}\\
			&\quad\text{for all}~r\in\Q_+~\text{with}~r<\textstyle
				\frac{1}{3k}\\
			&\quad\quad\text{there are}~x_1,\dots,x_m\in X_1\times 
				X_2~\text{such that}\\
			&\quad\quad\quad\quad D\in\textstyle\bigcap_{i=1}^m
				\Event(F\cap S_r(x_i),k_i,k) \text.
	\end{align*}\par\bigskip

	We show both directions.
	\begin{description}
		\item[$(\Rightarrow)$] Let $D\in\Event^\geq(F,n,k)$. Then $D$
			is $k$-coarse.  Suppose $D\in\Event(F,n',k)$ for some
			fixed $n'\geq n$ and let $f_1,\dots, f_m$ be the
			distinct facts from $F$ that appear in $D$ with $k_i$
			being their multiplicity in $D$ for $1\leq i\leq m$,
			also let $f_i=R_i (t_i^{(1)},t_i^{(2)})$ such
			that $t_i^{(1)}\in\dom_{\Schema} (R_1)$ and
			$t_i^{(2)}\in\dom_{\Schema}(R_2)$. Now since
			both $X_1$ and $X_2$ are dense, for every $1\leq i\leq
			m$ and all $r>0$ there are $x_i^{(1)}\in X_1$ and
			$x_i^{(2)}\in X_2$ such that $R_1
			(t_i^{(1)})\in B_r(x_i^{(1)})$ and
			$R_2(t_i^{(2)}) \in B_r(x_i^{(2)})$. By the
			choice of $r$ and since $D$ is $k$-coarse, for any two
			$i,j$ with $t_i^{(1)}\neq t_j^{(1)}$, it holds that
			$B_r(x_i^{(1)})$ and $B_r(x_j^{(1)})$ are disjoint (and
			similarly for the second part).  This means that
			$S_r(x_i)$ and $S_r(x_j)$ are disjoint for $i\neq j$,
			because in this case, $t_i^{(1)}\neq t_j^{(1)}$ or
			$t_i^{(2)}\neq t_j^{(2)}$.  Thus, $Q(D)$
			contains at most $1$ fact in each of the $S_r(x_i)$
			($1\leq i\leq m$). Thus $\hits_{Q(D)}(F\cap
			S_r(x_i))= \hits_{Q(D)} (F\cap\set{f_i}) =
			k_i$ for all $1\leq i\leq m$.
		\item[$(\Leftarrow)$] Now towards the other direction, one
			notices similarly to above that the $S_r(x_i)$ are
			pairwise disjoint (which follows from the
			$k$-coarseness of $D$ and the upper bound on $r$). This
			means that $D$ has at least $\sum_{i=1}^m k_i=n$ facts
			(including copies) in $\bigcup_i (F\cap
			S_r(x_i))\subseteq F$ and consequently
			$D\in\Event^\geq(F,n,k)$.\qedhere
	\end{description}
\end{subproof}

\par\bigskip
Now we are finally able to conclude the proof of
\cref{lem:crossprod-measurable-app}.\par\bigskip

\begin{proof}[\proofname~(\cref{lem:crossprod-measurable-app})]
	For every $F\in\Measurable[F]_{\Schema'}(R)$ and $n\in\N$,
	$Q^{-1}(\CEvent'(F,n))=\bigcup_{k\in \N}\Event(F,n,k)$ is
	measurable using \cref{clm:DFnk-measurable}. Since the events 
	$\CEvent'(F,n)$ generate $\Measurable'$, $Q$ is $(\Measurable,
	\Measurable')$-measurable.
\end{proof}

\section{Specifications of Considered Queries}\label{app:operations}

\begin{table}[H]
	\centering
	\caption[Specifications for relation base queries]
		{Specifications for relation base queries $Q=R$}
	\begin{tabular}{@{}lp{.7\textwidth}@{}}
		\toprule
		\textbf{\sffamily Prerequisites} & $R\in\Rel$\\
		\textbf{\sffamily Target Schema} & $R'=R$, keeping its
			type (in particular $\Att'$ is the set of attributes
			that appear in $\type_{\Schema}(R)$, keeping their
			domains and $\sigma$-algebras)\\
		\textbf{\sffamily Semantics} & $\hits_{Q(D)}(f) =
			\hits_D(f)$ for all facts $f\in\facts_{\Schema'}(\Rel')=
			\facts_{\Schema}(R)$\\
		\bottomrule
	\end{tabular}
\end{table}

\begin{table}[H]
	\centering
	\caption[Specifications for singleton base queries]
		{Specifications for singleton base queries $Q=\bag{R(a)}$}
	\begin{tabular}{@{}lp{.7\textwidth}@{}}
		\toprule
		\textbf{\sffamily Prerequisites} & $a$ belongs to some standard
			Borel space $(X,\Borel_X)$\\
		\textbf{\sffamily Target Schema} & $R'=R$ with
			$\type_{\Schema'}(R)=A'$ and $\Att'=\set{A'}$
			where $\dom_{\Schema'}(A')=X$\\
		\textbf{\sffamily Semantics} & $\hits_{Q(D)}(f) = 1$
			if $f=R(a)$ and $0$ otherwise\\
		\bottomrule
	\end{tabular}
\end{table}

\begin{table}[H]
	\centering
	\caption[Specifications for rename queries]
		{Specifications for rename queries $Q=\rename[{A\to{}B}](R)$}
	\begin{tabular}{@{}lp{.7\textwidth}@{}}
		\toprule
		\textbf{\sffamily Prerequisites} & $R\in\Rel$ such that
			$A$ appears in the type of $R$ but $B$
			does not\\
		\textbf{\sffamily Target Schema} & $R'=R$ and
			$\type_{\Schema'}(R)$ is obtained from
			$\type_{\Schema}(R)$ by replacing $A$ with 
			$B$; the set $\Att'$ consists of the attributes
			appearing in $\type_{\Schema'}(R)$ where $B$
			inherits its domain and $\sigma$-algebra from $A$
			\\
		\textbf{\sffamily Semantics} & $\hits_{Q(D)}(R(t))
			=\hits_D(R(t))$ for all $t\in\dom_{\Schema'}(R)$\\
	\bottomrule
	\end{tabular}
\end{table}

\begin{table}[H]
	\centering
	\caption[Specifications for additive union queries]
		{Specifications for additive union queries $Q=R_1\addunion{}R_2$}
	\begin{tabular}{@{}lp{.7\textwidth}@{}}
		\toprule
		\textbf{\sffamily Prerequisites} & $R_1,R_2\in\Rel$, both
			being of the same type\\
		\textbf{\sffamily Target Schema} & $\type_{\Schema'}(R')=
			\type_{\Schema}(R_1)=\type_{\Schema}(R_2)$; the
			set $\Att'$ consists of the attributes appearing in that
			type and they inherit their domains and
			$\sigma$-algebras\\
		\textbf{\sffamily Semantics} & $\hits_{Q(D)}
			(R'(t))=\hits_D(R_1(t))+\hits_D(R_2(t))$ for
			all $t\in\dom_{\Schema'}(R')$\\
		\bottomrule
	\end{tabular}
\end{table}

\begin{table}[H]
	\centering
	\caption[Specifications for min-intersection queries]
		{Specifications for min-intersection queries $Q=R_1\intersection{}R_2$}
	\begin{tabular}{@{}lp{.7\textwidth}@{}}
		\toprule
		\textbf{\sffamily Prerequisites} & $R_1,R_2\in\Rel$, both
			being of the same type\\
		\textbf{\sffamily Target Schema} & $\type_{\Schema'}(R')=
			\type_{\Schema}(R_1)=\type_{\Schema}(R_2)$; the
			set $\Att'$ consists of the attributes appearing in that
			type and they inherit their domains and
			$\sigma$-algebras\\
		\textbf{\sffamily Semantics} & $\hits_{Q(D)}
			(R'(t))=\min\set{\hits_D(R_1(t)),
			\hits_D(R_2(t))}$ for all $t\in\dom_{\Schema'}
			(R')$\\
		\bottomrule
	\end{tabular}
\end{table}

\begin{table}[H]
	\centering
	\caption[Specifications for difference queries]
		{Specifications for difference queries $Q=R_1\difference{}R_2$}
	\begin{tabular}{@{}lp{.7\textwidth}@{}}
		\toprule
		\textbf{\sffamily Prerequisites} & $R_1,R_2\in\Rel$, both
			being of the same type\\
		\textbf{\sffamily Target Schema} & $\type_{\Schema'}(R')=
			\type_{\Schema}(R_1)=\type_{\Schema}(R_2)$; the
			set $\Att'$ consists of the attributes appearing in that
			type and they inherit their domains and
			$\sigma$-algebras\\
		\textbf{\sffamily Semantics} & $\hits_{Q(D)}
			(R'(t))=\max\set{0,\hits_D(R_1(t))-\hits_D
			(R_2(t))}$ for all $t\in\dom_{\Schema'}(R')$\\
		\bottomrule
	\end{tabular}
\end{table}

\begin{table}[H]
	\centering
	\caption[Specifications for max-union queries]
		{Specifications for max-union queries $Q=R_1\maxunion{}R_2$}
	\begin{tabular}{@{}lp{.7\textwidth}@{}}
		\toprule
		\textbf{\sffamily Prerequisites} & $R_1,R_2\in\Rel$, both
			being of the same type\\
		\textbf{\sffamily Target Schema} & $\type_{\Schema'}(R')=
			\type_{\Schema}(R_1)=\type_{\Schema}(R_2)$; the
			set $\Att'$ consists of the attributes appearing in that
			type and they inherit their domains and
			$\sigma$-algebras\\
		\textbf{\sffamily Semantics} & $\hits_{Q(D)}
			(R'(t))=\max\set{\hits_D(R_1(t)),
			\hits_D(R_2(t))}$ for all $t\in\dom_{\Schema'}
			(R')$\\
		\bottomrule
	\end{tabular}
\end{table}

\begin{table}[H]
	\centering
	\caption[Specifications for deduplication queries]
		{Specifications for deduplication queries $Q=\dedupe(R)$}
	\begin{tabular}{@{}lp{.7\textwidth}@{}}
		\toprule
		\textbf{\sffamily Prerequisites} & $R\in\Rel$\\
		\textbf{\sffamily Target Schema} & $\type_{\Schema'}(R')=
			\type_{\Schema}(R)$; the set $\Att'$ consists of the
			attributes appearing in that type and they keep their
			domains and $\sigma$-algebras\\
		\textbf{\sffamily Semantics} & $\hits_{Q(D)}
			(f)=1$ if $\hits_D(f)>0$ and $0$ otherwise\\
		\bottomrule
	\end{tabular}
\end{table}

\begin{table}[H]
	\centering
	\caption[Specifications for selection queries (attribute value 
		equality)]
		{Specifications for selection queries $Q=\select[{A=
		B}](R)$}
	\begin{tabular}{@{}lp{.7\textwidth}@{}}
		\toprule
		\textbf{\sffamily Prerequisites} & $R\in\Rel$ and $A,
			B$ are distinct, comparable attributes from
			$\type_{\Schema}(R)$	\\
		\textbf{\sffamily Target Schema} & $R'=R$, keeping its
			type; and $\Att'$ is the restriction of $\Att$ to
			$\type_{\Schema}(R)$\\
		\textbf{\sffamily Semantics} & $\hits_{Q(D)}(f) =
			\hits_D(f)$ if $f_A=f_B$, and $0$ otherwise\\
	\bottomrule
	\end{tabular}
\end{table}

\begin{table}[H]
	\centering
	\caption[Specifications for selection queries (membership in Borel set)]
	{Specifications for selection queries $Q=\select[(A_1,\dots,A_k)\in
	\bm{B}](R)$}
	\begin{tabular}{@{}lp{.7\textwidth}@{}}
		\toprule
		\textbf{\sffamily Prerequisites} & $R\in\Rel$ and $A_1,\dots,
			A_k$ are mutually distinct attributes from 
			$\type_{\Schema}(R)$\\
		\textbf{\sffamily Target Schema} & $R'=R$, keeping its
			type; and $\Att'$ is the restriction of $\Att$ to
			$\type_{\Schema}(R)$\\
		\textbf{\sffamily Semantics} & $\hits_{Q(D)}(f) =
		\hits_D(f)$ if $f_{A_1,\dots,A_k}\in\bm{B}$, and $0$
			otherwise\\
	\bottomrule
	\end{tabular}
\end{table}

\begin{table}[H]
	\centering
	\caption[Specifications for projection queries]
	{Specifications for projection queries $Q=\project[A_1,\dots,A_k](R)$}
	\begin{tabular}{@{}lp{.7\textwidth}@{}}
		\toprule
		\textbf{\sffamily Prerequisites} & $R\in\Rel$ and $A_1,\dots,
			A_k$ are mutually distinct attributes from 
			$\type_{\Schema}(R)$\\
		\textbf{\sffamily Target Schema} & $R'=R$ with $\type_{\Schema'}
			(R) = (A_1,\dots, A_k)$ and $\Att'=\set{A_1,\dots,A_k}$
			with domains and $\sigma$-algebras inherited from $\Att$
			\\
		\textbf{\sffamily Semantics} & $\hits_{Q(D)}(f') =
		\hits_D(\set{f\in\facts_{\Schema}(R)\colon f'_{A_1,\dots,
		A_k}=f_{A_1,\dots,A_k}})$\\
		\bottomrule
	\end{tabular}
\end{table}

\begin{table}[H]
	\centering
	\caption[Specifications for cross product queries]
	{Specifications for cross product queries $Q=R_1\product R_2$}
	\begin{tabular}{@{}lp{.7\textwidth}@{}}
		\toprule
		\textbf{\sffamily Prerequisites} & $R_1,R_2\in\Rel$\\
		\textbf{\sffamily Target Schema} & $R'=R$ with $\type_{\Schema'}
			(R)=\type_{\Schema}(R_1)\times\type_{\Schema}(R_2)$, and
			$\Att'$ being the attributes from $\type_{\Schema}(R_1)
			\cup\type_{\Schema}(R_2)$, inheriting domains and
			$\sigma$-algebras\\
		\textbf{\sffamily Semantics} & $\hits_{Q(D)}(R(t_1,
			t_2)) = \hits_D(R_1(t_1))\cdot\hits_D(R_2(t_2))$ where
			$t_1\in\dom_{\Schema}(R_1)$ and $t_2\in\dom_{\Schema}
			(R_2)$\\		
		\bottomrule
	\end{tabular}
\end{table}

\begin{table}[H]
	\centering
	\caption[Specifications for natural join queries]
	{Specifications for natural join queries $Q=R_1\natjoin R_2$}
	\begin{tabular}{@{}lp{.7\textwidth}@{}}
		\toprule
		\textbf{\sffamily Prerequisites} & $R_1,R_2\in\Rel$\\
		\textbf{\sffamily Target Schema} & $R'=R$ with $\type_{\Schema'}
			(R)=\type_{\Schema}(R_1)\cup\type_{\Schema}(R_2)$, and
			$\Att'$ being the attributes from $\type_{\Schema}(R_1)
			\cup\type_{\Schema}(R_2)$, inheriting domains and
			$\sigma$-algebras\\
		\textbf{\sffamily Semantics} & $\hits_{Q(D)}(R(t)) =
		\hits_D(R_1(t_{\type_{\Schema}(R_1)}))\cdot
		\hits_D(R_2(t_{\type_{\Schema}(R_2)}))$\\		
		\bottomrule
	\end{tabular}
\end{table}

\begin{table}[H]
	\centering
	\caption[Specifications for aggregate queries]
	{Specifications for aggregate queries $Q=\varpi_{A_1,\dots,A_k,\Phi(A)}(R)$}
	\begin{tabular}{@{}lp{.7\textwidth}@{}}
		\toprule
		\textbf{\sffamily Prerequisites} & $R\in\Rel$, $A,A_1,\dots,A_k
			\in\type_{\Schema}(R)$ with $A_i\neq A_j$ for $i\neq j$,
			and $\Phi$ being an aggregator from $U$ to $V$ with
			$U=\dom_{\Schema}(A)$\\
		\textbf{\sffamily Target Schema} & $\type_{\Schema'}(R')=(A_1,
			\dots,A_k,A')$ where $A_1,\dots, A_k$ keep their domains
			and $\sigma$-algebra and $\dom_{\Schema'}(A')=V$ with
			its inherent $\sigma$-algebra\\
		\textbf{\sffamily Semantics} & $\hits_{Q(D)}
			(R'(a_1,\dots,a_k,c)) = 1$ if $\Phi(D')=c$ and $0$
			otherwise; hereby, $D' = \big(\project[A]\circ
	\select[A_1=a_1,\dots,A_k=a_k](R)\big)(D)$\\
		\bottomrule
	\end{tabular}
\end{table}

\end{document}